\documentclass[11pt,a4paper]{article}
\usepackage{geometry}
\usepackage{graphicx}
\usepackage{bm,array}
\usepackage{comment}
\usepackage{caption}
\usepackage{subcaption}
\usepackage{tablefootnote}
 \usepackage{makecell}
\usepackage{booktabs}
 \usepackage{multirow}
 \usepackage{float}

\usepackage[normalem]{ulem}

\usepackage[pdfpagemode={UseOutlines},bookmarks=true,bookmarksopen=true,
bookmarksopenlevel=0,bookmarksnumbered=true,hypertexnames=false,
colorlinks,linkcolor={blue},citecolor={blue},urlcolor={red},
pdfstartview={FitV},unicode,breaklinks=true]{hyperref}
\hypersetup{urlcolor=blue, colorlinks=true}

\usepackage{setspace}
\usepackage{vmargin}
\setmarginsrb{ 1.0in}  
{ 0.6in}  
{ 1.0in}  
{ 0.8in}  
{  20pt}  
{0.25in}  
{9pt}  
{ 0.3in}  

\usepackage{authblk}

\usepackage{amsmath}
\usepackage{amsfonts}
\usepackage{amssymb}
\usepackage[round, sort,comma,authoryear]{natbib}

\usepackage{tikz}
\usetikzlibrary{arrows.meta,positioning,shapes,fit}

\newtheorem{theorem}{Theorem}
\newtheorem{assumption}{Assumption}

\newtheorem{lemma}{Lemma}

\newtheorem{proposition}{Proposition}
\newtheorem{remark}{Remark}

\newenvironment{proof}[1][Proof]{\noindent\textbf{#1.} }{\ \rule{0.5em}{0.5em}}

\usepackage{multirow}
\usepackage{hhline}
\usepackage{footnote}
\usepackage[ruled,vlined]{algorithm2e}
\usepackage{algorithmic}

\usepackage{tcolorbox}
\tcbuselibrary{breakable} 

\usepackage{pgfplots}
\pgfplotsset{compat=1.18}

\newcommand{\transpose}{{\mbox{\tiny T}}}

\newcommand{\cA}{{\mathcal{A}}}

\newcommand{\cS}{{\mathcal{S}}}

\newcommand{\cN}{{\mathcal{N}}}

\newcommand{\cT}{{\mathcal{T}}}

\newcommand{\bQ}{\textbf{Q}}

\newcommand{\bu}{\textbf{u}}

\newcommand{\bV}{\textbf{V}}

\newcommand{\br}{\textbf{r}}
 
\newcommand{\bbt}{\pmb{\beta}}

\newcommand{\bmu}{{\pmb{\mu}}}

\newcommand{\bbE}{\mathbb{E}}

\usepackage{mathtools}

\usepackage{mathtools}

\newif\ifnotes\notestrue
\def\mtien#1{{\color{magenta}{#1}}}

\def\htien#1{}

\usepackage{xcolor}

\begin{document}







\newcolumntype{C}{>{\centering\arraybackslash}p{4em}}

\title{\textbf{
Equilibrium-Constrained Estimation of Recursive Logit Choice Models}}
\author[1,*]{Hung Tran}
\author[1]{Tien Mai}
\author[2]{Minh Hoang Ha}
\affil[1]{\it\small
School of Computing and Information Systems, Singapore Management University}
\affil[2]{\it\small
SLSCM and CADA, Faculty of Data Science and Artificial Intelligence, College of Technology, National Economics University, Hanoi, Vietnam}
\affil[*]{\it\small Corresponding author, hh.tran.2024@phdcs.smu.edu.sg}

\maketitle

\begin{abstract}
The recursive logit (RL) model provides a flexible framework for modeling sequential decision-making in transportation and choice networks, with important applications in route choice analysis, multiple discrete choice problems, and activity-based travel demand modeling. Despite its versatility, estimation of the RL model typically relies on nested fixed-point (NFXP) algorithms that are computationally expensive and prone to numerical instability. We propose a new approach that reformulates the maximum likelihood estimation problem as an optimization problem with equilibrium constraints, where both the structural parameters and the value functions are treated as decision variables. We further show that this formulation can be equivalently transformed into a conic optimization problem with exponential cones, enabling efficient solution using modern conic solvers such as MOSEK. Experiments on synthetic and real-world datasets demonstrate that our convex reformulation achieves accuracy comparable to traditional methods while offering significant improvements in computational stability and efficiency, thereby providing a practical and scalable alternative for recursive logit model estimation.
\end{abstract}




\section{Introduction}

The recursive logit (RL) model, introduced by \citet{FosgFrejKarl13}, has emerged as a powerful framework for modeling sequential decision-making processes on networks. 
In transportation research, the RL model has been widely applied to analyze route choice behavior, capturing how travelers make link-by-link decisions that ultimately determine complete paths. 
Beyond route choice, RL-type formulations have proven useful in activity-based travel demand modeling, where individuals plan sequences of activities throughout the day \citep{zimmermann2018mixedRL}, as well as in multiple discrete choice settings, where selections of composite alternatives can be represented as paths on suitably defined directed acyclic graphs \citep{tran2024network}. 
By decomposing complex combinatorial choice problems into recursive structures, the RL model provides both theoretical elegance and practical tractability, making it a cornerstone in the field of behavioral modeling and prediction.

Despite these advantages, the practical application of RL models critically depends on estimation methods. 
Existing approaches overwhelmingly rely on the nested fixed-point (NFXP) algorithm \citep{Rust87,Rust94}, in which model parameters are estimated by repeatedly solving a dynamic program to evaluate the likelihood function. 
In practice, this requires computing the value function as the fixed point of a Bellman operator through iterative methods such as value iteration. 
However, this procedure is computationally demanding and often numerically unstable. 
When the Bellman recursion is ill-conditioned, value iteration can fail to converge, diverging instead to infinity or producing undefined values. 
As a result, the computed log-likelihood may become invalid, leading to failed estimation runs or unreliable parameter estimates~\citep{FosgFrejKarl13,MaiFrejinger2022}. These limitations represent a serious bottleneck for large-scale applications, where estimation speed and numerical robustness are critical.

To address these challenges, we propose a novel reformulation of the RL maximum likelihood estimation problem. 
Our key idea is to represent the likelihood estimation task as an optimization problem with equilibrium constraints, in which both the structural parameters and the expected downstream utilities (i.e., value functions) are treated as joint decision variables. 
This formulation allows us to transform the problem into a convex conic program, specifically an exponential-cone optimization model, that can be solved efficiently by modern commercial solvers such as MOSEK \citep{mosek2023}. 
By avoiding repeated fixed-point computations and instead embedding equilibrium conditions directly into the optimization problem, our approach provides a robust and scalable alternative to classical NFXP.

\paragraph{Contributions.} 
This paper makes several methodological and practical contributions to the literature on RL model estimation:

\begin{enumerate}
    \item \textbf{A new equilibrium-constrained reformulation.}  
    We introduce, for the first time, an equilibrium-constrained reformulation of the recursive logit maximum likelihood estimation problem.  
    Unlike the classical NFXP approach, which treats the value functions as implicit fixed points that must be repeatedly computed, our formulation elevates them to optimization variables and imposes Bellman consistency conditions explicitly as constraints.  
    This reformulation establishes a direct optimization view of RL estimation, bridging dynamic discrete choice modeling and conic optimization.

    \item \textbf{Exact convexification via exponential-cone programming.}  
    We prove that the equilibrium-constrained problem can be equivalently reformulated as a convex program using exponential-cone constraints.  
    This result is highly non-trivial: it shows that a problem traditionally considered non-convex and difficult to solve admits an exact convex representation, ensuring global optimality.  
    By leveraging modern interior-point solvers (e.g., MOSEK), the estimation task can now be solved with polynomial-time complexity and numerical robustness, circumventing the instability of value iteration and the local-convergence issues of NFXP.

    \item \textbf{A scalable trimming procedure for large-scale networks.}  
    To address the curse of dimensionality in real-world applications, we design a principled network-trimming method based on flow dominance.  
    This procedure automatically prunes states and arcs that contribute negligibly to the likelihood, while rigorously preserving connectivity and statistical validity.  
    The trimming ensures that the convex program remains tractable without compromising estimation accuracy, thereby extending the applicability of RL estimation to large, dense networks.


    \item \textbf{Comprehensive empirical validation.}  
    Through extensive experiments on synthetic benchmarks and real-world datasets (including multiple discrete choice and route choice applications), we demonstrate that our approach consistently matches the predictive accuracy of NFXP while offering vastly superior robustness.  
    In particular, our method avoids the convergence failures and invalid likelihood evaluations that frequently plague NFXP, while often achieving order-of-magnitude speedups in moderate-sized problems.  
    These results establish convex conic optimization as a practical and reliable new paradigm for RL estimation.

\end{enumerate}

Overall, this work is the first to show that the estimation the RL model can be cast as an \emph{exact convex optimization problem}, thereby replacing decades-old reliance on unstable fixed-point iterations with a theoretically principled and computationally efficient alternative. 
This shift has both methodological significance—connecting dynamic discrete choice estimation with conic optimization—and practical impact, enabling stable large-scale estimation in domains where RL models have so far been under-utilized.

\paragraph{Paper Outline:} The remainder of this paper is organized as follows. 
Section~\ref{sec:review} reviews relevant literature on route choice modeling, recursive logit estimation methods, and related optimization-based approaches. 
Section~\ref{sec:RL} introduces the RL framework and revisits its estimation using the classical NFXP algorithm. 
In Section~\ref{sec:ECP}, we present the proposed equilibrium-constrained  reformulation and discuss its theoretical properties, including convexity and equivalence to the standard RL formulation. 
Section~\ref{sec:experiments} reports numerical experiments on both synthetic and real transportation networks, highlighting the computational efficiency and stability of the proposed approach compared to NFXP. 
Finally, Section~\ref{sec:concl} concludes the paper and discusses potential directions for future research. Appendix extends the framework to the Nested Recursive Logit model \citep{MaiFosFre15}. 

\section{Literature Review}\label{sec:review}
Route choice models play a central role in transportation research and travel behavior analysis, with applications ranging from network design and congestion management to demand forecasting and intelligent transportation systems~\citep{BenABier99,Trai03,Domencich1975}.
The literature on route choice modeling can broadly be categorized into two main approaches: path-based and link-based recursive formulations. Path-based models~\citep[see][for a review]{Prat09} rely on sampling feasible paths between each origin--destination  pair, which makes the resulting parameter estimates sensitive to the path sampling procedure. Even when correction terms are introduced to ensure consistent estimation, prediction remains computationally demanding, particularly in large-scale networks where the number of feasible paths grows exponentially.

In contrast, link-based RL models~\citep{Fosgerau2013,MaiFosFre15,mai2021RL_STD,oyama2017discounted} build upon the dynamic discrete choice  framework of~\citep{Rust87} and are mathematically equivalent to discrete choice models defined over the set of all feasible paths. These models offer several advantages: they can be consistently estimated without path enumeration, and they enable efficient prediction via dynamic programming. The first RL model was introduced by~\citet{Fosgerau2013} and has since been extended in multiple directions, including modeling correlated utilities across overlapping paths~\citep{MaiFosFre15,Mai_RNMEV,mai2018decomposition}, accounting for dynamic network conditions~\citep{de2020RL_dynamic}, incorporating stochastic time-dependent link costs~\citep{mai2021RL_STD},  capturing discounted behavior~\citep{oyama2017discounted}, or modeling multiple-discrete choice behaviors \citep{tran2024network}. 
RL models have been successfully applied in various transportation and network design contexts, including traffic management~\citep{BailComi08,Melo12}, network pricing~\citep{zimmermann2021strategic}, and network interdiction~\citep{mai2024stackelberg}. A comprehensive overview of these developments and related estimation techniques is provided in~\citet{zimmermann2020tutorial}.



The RL model belongs to the class of \emph{dynamic discrete choice}  models, where estimation is typically performed using the NFXP algorithm~\citep{Rust87}. In this framework, the inner loop solves the Bellman equation to obtain the value function, while the outer loop maximizes the likelihood with respect to model parameters. Despite its conceptual simplicity, several studies have noted that NFXP is highly sensitive to initialization, scaling, and network topology, which often lead to numerical instability or infeasible value function estimates in practice~\citep{Fosgerau2013,mai2018decomposition,MaiFrejinger2022}. 
 Several efforts have been made to improve the NFXP procedure for RL model estimation. For instance,~\citet{mai2018decomposition} proposed a decomposition-based method that can reduce estimation time by up to a factor of thirty. However, this approach cannot accommodate individual-specific attributes, which limits its applicability in general route choice settings. More recently,~\citet{MaiFrejinger2022} introduced a least-squares-based technique to accelerate the computation of the value function within the NFXP framework. Despite its computational efficiency, this method suffers from the same numerical instability as standard NFXP—specifically, the value function may fail to converge during the parameter search process, leading to incomplete or invalid likelihood evaluations.
These challenges motivate the search for alternative estimation frameworks that preserve the theoretical rigor of the RL model while improving numerical robustness.

Our work contributes to this line of research by connecting route choice estimation with the field of \emph{conic optimization}. Specifically, we propose an \emph{equilibrium-constrained optimization}  formulation that casts the RL estimation problem as a convex conic program. Compared to standard gradient-based optimization approaches, conic formulations offer several advantages, including guaranteed global optimality under convexity, well-established interior-point solution methods, and improved numerical stability in ill-conditioned problems~\citep{BoydVandenberghe2004,nesterov2006conic,mosek2023}. \textit{This approach bridges dynamic discrete choice modeling and modern convex optimization, providing a new computational perspective for stable and scalable estimation of RL models.}

It is worth emphasizing that our proposed equilibrium-constrained formulation framework is conceptually related to the Mathematical Programming with Equilibrium Constraints (MPEC) approaches developed for structural and dynamic discrete choice estimation~\citep{SuJudd2012,IskhakovEtAl2016,qeconomics2023}. Despite this connection, our method differs fundamentally in both formulation and scope.  MPEC methods jointly estimate model parameters and value functions by imposing the Bellman equation as an equilibrium constraint. However, these formulations typically preserve the nonlinear and nonconvex structure of the original problem, resulting in large-scale constrained nonlinear programs that are computationally demanding and often difficult to solve to global optimality. In contrast, our convex  formulation leverages the specific analytical structure of the RL model to derive a conic representation of the equilibrium conditions. This convexification guarantees global optimality, enhances numerical stability, and enables the use of efficient interior-point solvers such as MOSEK. 
Moreover, as a side note, prior studies have shown that MPEC approaches are generally less computationally efficient than traditional NFXP algorithms, particularly in large-scale settings where the number of equilibrium constraints increases rapidly with the state space~\citep{IskhakovEtAl2016}. In contrast, as shown extensively in our experiments, our proposed method achieves both numerical robustness and computational tractability, offering a practical and theoretically grounded alternative to the NFXP.

\section{The Recursive Logit Model}\label{sec:RL}

The \emph{recursive logit} (RL) model, introduced by \citet{FosgFrejKarl13}, is a dynamic discrete choice framework designed to represent sequential decision-making on a transportation network. Unlike classical multinomial logit formulations that treat entire paths as atomic alternatives, the RL model decomposes path choice into a sequence of link- or node-level decisions, thereby avoiding the combinatorial explosion of path enumeration.

Let $\cS$ denote the set of states, where each state $s \in \cS$ corresponds to a link or node in the network. For each state $s$, let $A(s) \subseteq \cS$ represent the set of directly reachable successor states. For a trip with destination $d$, the destination is modeled as an absorbing state, so the effective state space is $\widetilde{\cS} = \cS \cup \{d\}$. At each step, a traveler located in state $s$ chooses a successor state $s' \in A(s)$. The instantaneous random utility associated with this transition is given by
\begin{equation}
    u(s'|s) = v(s'|s) + \mu \, \epsilon(s'),
\end{equation}
where $v(s'|s)$ denotes the deterministic component of utility, typically expressed as a linear function of observable attributes (e.g., travel time, cost, or congestion) with parameters to be estimated, $\boldsymbol{\beta}$. Hence, we can write this component as $v(s'|s;~\boldsymbol{\beta})$. The term $\epsilon(s')$ represents an i.i.d.\ extreme value type~I disturbance, and $\mu > 0$ is a scale parameter.
At each state $s \in \cS$, the individual chooses the next state by maximizing the sum of two components: 
(i) the \emph{instantaneous random utility} of moving from $s$ to $s' \in A(s)$, and 
(ii) the \emph{expected maximum utility} of continuing from $s'$ to the destination $d$. 
Formally, the next-state decision is characterized by
\begin{equation}
    \max_{s' \in A(s)} \left\{ u(s'|s) + V(s') \right\},
\end{equation}
where $V(s')$ denotes the expected maximum utility of reaching $d$ starting from state $s'$. By the dynamic programming principle, the value function $V(s)$ can be expressed as
\begin{equation}
    V(s) =
    \begin{cases}
        0, & s = d, \\[0.5em]
        \bbE_{\epsilon}\!\left[ \max_{s' \in A(s)} \big\{ v(s'|s) + V(s') + \mu \, \epsilon(s') \big\} \right], & s \in \cS,
    \end{cases}
    \label{eq:dpvalue}
\end{equation}
where the expectation is taken with respect to the random shocks $\{\epsilon(s')\}_{s' \in A(s)}$.

In the RL model, the error terms $\epsilon(s')$ are assumed to be i.i.d.\ extreme value type I (Gumbel). Under this distributional assumption, the maximum operator in \eqref{eq:dpvalue} admits a closed-form representation, and the resulting conditional choice probabilities follow a multinomial logit structure. Consequently, the value function satisfies the \emph{log-sum-exp} recursion:
\begin{equation}
    V(s) =
    \begin{cases}
        0, & s = d, \\[0.5em]
        \mu \ln \left( \sum_{s' \in A(s)} \exp\!\left(\tfrac{1}{\mu} \left( v(s'|s) + V(s') \right) \right) \right), & s \neq d.
    \end{cases}
    \label{eq:valuefunction}
\end{equation}
Equation \eqref{eq:valuefunction} illustrates how the RL model inherits the tractability of the multinomial logit: the continuation value at each state is obtained as a smooth ``log-sum-exp'' operator of the utilities of feasible successor states. This recursive structure eliminates the need to enumerate paths explicitly and ensures that choice probabilities are consistent with random utility maximization.

Moreover, conditional on being at state $s$, the probability of choosing successor $s' \in A(s)$ is
\begin{equation}
    P(s'|s) = 
    \frac{\exp\!\big(\frac{1}{\mu}(v(s'|s) + V(s'))\big)}{\sum_{t \in A(s)} \exp\!\big(\frac{1}{\mu}(v(t|s) + V(t))\big)}.
    \label{eq:choiceprob}
\end{equation}
This expression illustrates the recursive nature of the model: the attractiveness of an action depends not only on its immediate attributes but also on the continuation value $V(s')$.

For model estimation, consider an observed path 
$\sigma = \{s_0, s_1, \ldots, s_T\},$
which represents a sequence of states from an origin $s_0$ to the destination $d$.  
The probability of observing this path under parameter vector $\bbt$ is the product of the conditional choice probabilities at each step:
\begin{equation}
    P(\sigma \mid \bbt) 
    = \prod_{t=0}^{T-1} P(s_{t+1} \mid s_t).
\end{equation}
Using the recursive logit structure, this probability can be expressed in closed form as
\begin{equation}
    P(\sigma \mid \bbt) 
    = \exp\!\left( \frac{1}{\mu}(v(\sigma) - V(s_0)) \right),
\end{equation}
where 
$v(\sigma) = \sum_{t=0}^{T-1} v(s_{t+1} \mid s_t)$
is the total deterministic utility accumulated along the path $\sigma$, and $V(s_0)$ is the value function at the origin state.
Given a dataset of $N$ observed paths $\{\sigma_n\}_{n=1}^N$, the log-likelihood function for parameter estimation is
\begin{equation}
    \mathcal{L}(\bbt) 
    = \sum_{n=1}^N \log P(\sigma_n \mid \bbt)
    = \sum_{n=1}^N \Big( \frac{1}{\mu}(v(\sigma_n) - V^n(s_0^n)) \Big),
    \label{eq:loglikelihood}
\end{equation}
where $s_0^n$ denotes the origin of the $n$-th observed path, and $V^n(s)$ represents the expected maximum utility of reaching the destination starting from state $s$, corresponding to observation $n$. The model parameters $\bbt$ can then be estimated by \emph{maximum likelihood estimation} (MLE), which consists of maximizing the log-likelihood function $\mathcal{L}(\bbt)$ with respect to $\bbt$.

\paragraph{Estimation via NFXP.}
A central difficulty in maximizing the log-likelihood function in  \eqref{eq:loglikelihood} lies in the fact that the value functions $\{V(s)\}$ are not given in closed form but are instead implicitly defined through the fixed-point system \eqref{eq:valuefunction}. For a given parameter vector $\bbt$, solving this system may not always succeed: depending on the parameter values, the fixed-point equations can be ill-conditioned, which may lead to numerical instability or even the absence of a valid solution. This lack of robustness makes the computation of $V(s)$ itself a non-trivial task.

The standard estimation strategy is the NFXP algorithm \citep{Rust87}. In this approach, the estimation problem is organized into two loops. In the outer loop, a candidate parameter vector $\bbt$ is proposed, and the corresponding likelihood $\mathcal{L}(\bbt)$ must be evaluated. In the inner loop, for the given $\bbt$, the system \eqref{eq:valuefunction} is solved to recover the value functions $\{V(s)\}$, typically through iterative methods such as value iteration. Once $V(s)$ is obtained, the outer loop updates $\bbt$ using a numerical optimization routine, for example quasi-Newton methods. 

While this nested structure guarantees statistical consistency, it comes with a significant computational burden. Each evaluation of the likelihood requires repeatedly solving a high-dimensional dynamic program defined over the entire network. When this process must be performed many times across successive updates of $\bbt$, the computational cost grows prohibitively large. This combination of potential instability in computing value functions and the repeated fixed-point evaluations forms the main challenge of applying the RL model in practice. These difficulties motivate the search for alternative formulations, such as the equilibrium-constrained optimization approach developed in the next section.

\section{MLE as Equilibrium-constrained Optimization}\label{sec:ECP}
We now reformulate the maximum likelihood estimation problem as an optimization problem with equilibrium constraints. The key idea is to treat both the model parameters $\bbt$ and the value functions $\{V^n(s)\}$ as decision variables to be jointly estimated. For notational convenience, we denote the value function $V^n(s)$ by $V^n_s$, moving the state index to a subscript. The destination $d$ also depends on the observed path $n$, but the path index is omitted for clarity. This representation allows us to handle $\bbt$ and $V$ more naturally as vectors of optimization variables in the reformulated problem.

We write the MLE as the following equilibrium-constrained optimization problem:
\begin{align}
    \max_{\bbt,\, \bV} \quad & \sum_{n \in [N]} \Big( v(\sigma_n \mid \bbt) - V^n_{s^n_0} \Big) 
    \label{prob:mle-ec}\tag{\sf MLE-EC-EQ}\\
    \text{s.t.} \quad & 
    V^n_s = \log\!\left(\sum_{s' \in A(s)} \exp\!\big(v({s'|s;~\bbt}) + V^n_{s'}\big)\right), 
    \quad \forall n \in [N],~ s \in \cS. \nonumber\\
    &V^n_d = 0,\quad \forall n \in [N],\nonumber
\end{align}
In this formulation, the equilibrium conditions defining the  value functions appear explicitly as equality constraints. The goal is therefore to simultaneously optimize over the structural parameters $\bbt$ and the value functions $\{V^n_s\}$. However, problem \eqref{prob:mle-ec} is highly impractical in its current form: the constraints are nonlinear equalities involving exponential terms, which makes the optimization problem non-convex and computationally challenging.  In the next section, we describe our approach to reformulate \eqref{prob:mle-ec} into a tractable \emph{convex} optimization problem. This convex reformulation allows the problem to be efficiently solved using modern commercial solvers.

\subsection{Inequality-Constrained Equilibrium Reformulation}
We now describe our approach to reformulate the non-convex equilibrium-constrained problem \eqref{prob:mle-ec} into a convex optimization problem. The central idea is to relax the nonlinear equality constraints into inequalities that preserve the equilibrium structure while ensuring tractability. To make this reformulation valid, we impose two mild assumptions.
\begin{assumption}[A1 - Network connectivity]\label{assump:A1}
The state network is fully connected with respect to the origins appearing in the observed paths. In other words, for any observation $n$ with origin state $s^n_0$, and for any state $s \in \cS$, there exists at least one feasible path that connects $s^n_0$ to $s$.
\end{assumption}
Assumption~\ref{assump:A1} guarantees that all states are reachable from the observed origins. This condition is natural in most transportation or choice network applications, where the network is typically strongly connected or can be restricted to its reachable component.
\begin{assumption}[A2 - Existence of a feasible solution]\label{assump:A2}
The maximum likelihood estimation problem \eqref{prob:mle-ec} admits at least one feasible solution $(\bbt, \bV)$.
\end{assumption}
Assumption~\ref{assump:A2} ensures that the estimation problem is well-posed, i.e., the optimization problem has a feasible solution that maximizes the likelihood. This is a standard regularity condition in maximum likelihood estimation: without it, the optimization could be ill-defined (for instance, if the likelihood is unbounded or the feasible region is empty). In practice, this assumption is usually satisfied under mild conditions on the utility specification and the observation data.

The following theorem establishes a key result that enables us to overcome the nonlinearity of \eqref{prob:mle-ec}: it shows that the equality constraints in \eqref{prob:mle-ec} can be safely relaxed to inequalities  without changing the optimal solution. This relaxation not only preserves the theoretical integrity of the model but also paves the way for a convex reformulation 
that can be efficiently solved using modern optimization techniques.

\begin{theorem}
\label{th:ineq-relax}
Under Assumptions~\ref{assump:A1}--\ref{assump:A2}, the equality-constrained problem \eqref{prob:mle-ec} is equivalent to the inequality-relaxed problem:
\begin{align}
    \max_{\bbt,\,\bV} \quad & \sum_{n \in [N]} \Big( v(\sigma_n \mid \bbt) - V^n_{s^n_0} \Big) 
    \tag{\sf MLE-EC-IEQ} \label{prob:mle-ec-ineq}\\
    \text{s.t.} \quad & V^n_s \;\geq\; 
    \log\!\left(\sum_{s' \in A(s)} \exp\!\big(v(s'|s;\bbt) + V^n_{s'} \big)\right), 
    \quad \forall n \in [N], ~ s \in \cS, \nonumber \\
    & V^n_d = 0,\quad  \forall n \in [N]. \nonumber
\end{align}
\end{theorem}
For notational convenience, let us first define the Bellman operator $\mathcal{T}_{\boldsymbol{\beta}}$ acting on a vector of values $\boldsymbol{V}^n$ as
\[
(\mathcal{T}_{\boldsymbol{\beta}}[\boldsymbol{V}^n])_s =
\begin{cases}
\log\!\left(\displaystyle\sum_{s' \in A(s)} \exp\!\big(v(s' \mid s; \boldsymbol{\beta}) + V^n_{s'}\big)\right), & \text{if } s \neq d,\\[8pt]
0, & \text{if } s = d.
\end{cases}
\]
To prove the theorem, we first introduce the following lemma, which shows that the Bellman operator $\mathcal{T}_{\boldsymbol{\beta}}[\boldsymbol{V}^n]$ is monotonic in $\boldsymbol{V}^n$. This property is critical for relaxing the equality constraint in the theorem.
\begin{lemma}\label{lm:monotone}
The operator $\cT_{\bbt}$ is monotone: if $\bV^{(1)}\geq \bV^{(2)}$ component-wise, then $\cT_{\bbt}[\bV^{(1)}]\geq \cT_{\bbt}[\bV^{(2)}]$. Moreover, if $V^{(1)}_s>V^{(2)}_s$ for some $s$, then $\cT_{\bbt}[\bV^{(1)}]_{s'} > \cT_{\bbt}[\bV^{(2)}]_{s'}$ for every predecessor $s'\in D(s)$, where $D(s)=\{s'\in\cS : s\in A(s')\}$.
\end{lemma}

\begin{proof}[Proof of Lemma]
For each $s\in\cS$, define the function
\[
F_s(\mathbf{z}) \;:=\; \log\!\left(\sum_{a\in A(s)} \exp(z_a)\right),
\qquad \mathbf{z}=(z_a)_{a\in A(s)}\in\mathbb{R}^{A(s)}.
\]
It is standard that $F_s$ is (i) non-decreasing in each argument and (ii) \emph{strictly} increasing in any argument for which the input strictly increases, because the map $\mathbf{z}\mapsto \sum_{a} e^{z_a}$ is strictly increasing in each coordinate and $\log(\cdot)$ is strictly increasing on $(0,\infty)$.

To prove  the monotonicity, assume $\bV^{(1)}\ge \bV^{(2)}$ on $\cS$. Fix any $s\in\cS$ and set, for $i\in\{1,2\}$,
\[
\mathbf{z}^{(i)}_s \;:=\; \big(v(a|s;\bbt)+V^{(i)}_a\big)_{a\in A(s)}.
\]
Then $\mathbf{z}^{(1)}_s \ge \mathbf{z}^{(2)}_s$ component-wise, hence by the monotonicity of $F_s$,
\[
\big(\cT_{\bbt}[\bV^{(1)}]\big)_s
= F_s\!\big(\mathbf{z}^{(1)}_s\big)
\;\ge\;
F_s\!\big(\mathbf{z}^{(2)}_s\big)
= \big(\cT_{\bbt}[\bV^{(2)}]\big)_s.
\]
Since $s\in\cS$ was arbitrary, the inequality holds component-wise on $\cS$.

Moreover, let us suppose that there exists $s\in\cS$ with $V^{(1)}_s>V^{(2)}_s$. Let $s'\in D(s)$, i.e., $s\in A(s')$. For $i\in\{1,2\}$,  define
\[
\mathbf{z}^{(i)}_{s'} \;:=\; \big(v(a|s';\bbt)+V^{(i)}_a\big)_{a\in A(s')}.
\]
Then for the coordinate corresponding to $a=s$ we have
\[
v(s|s';\bbt)+V^{(1)}_s \;>\; v(s|s';\bbt)+V^{(2)}_s.
\]
Hence $\mathbf{z}^{(1)}_{s'}$ strictly dominates $\mathbf{z}^{(2)}_{s'}$ in at least that coordinate, and by the \emph{strict} monotonicity of $F_{s'}$ in each argument,
\[
\big(\cT_{\bbt}[\bV^{(1)}]\big)_{s'}
= F_{s'}\!\big(\mathbf{z}^{(1)}_{s'}\big)
\;>\;
F_{s'}\!\big(\mathbf{z}^{(2)}_{s'}\big)
= \big(\cT_{\bbt}[\bV^{(2)}]\big)_{s'}.
\]
This holds for every predecessor $s'\in D(s)$, which proves the claim.
\end{proof}

We now return to the proof of Theorem~\ref{th:ineq-relax}. 
The general idea is to show that the equality constraint in the original formulation can be relaxed to an inequality constraint without altering the optimal solution. 
This is achieved by leveraging the monotonicity of the Bellman operator $\mathcal{T}_{\boldsymbol{\beta}}$, established in Lemma~\ref{lm:monotone}, which ensures that iterative applications of $\mathcal{T}_{\boldsymbol{\beta}}$ preserve the order of value functions.  Consequently, any feasible solution satisfying the relaxed inequality will converge to the same fixed point as that of the equality-constrained formulation.

\begin{proof}[Proof of the Theorem]
The idea of the proof is to show that relaxing the equalities to inequalities does not enlarge the optimal solution set. Although the feasible region becomes bigger, any optimal solution of the relaxed problem must in fact satisfy the equalities. Intuitively, this is because any slack in the inequalities can be ``tightened'' to produce a strictly better objective value, contradicting optimality. For clarity of exposition, we structure the proof into the following steps.

\medskip
\noindent\textit{Step 1. Equality implies inequality:}   
It is immediate that if any solution $(\bbt,\bV)$ satisfies the equalities of \eqref{prob:mle-ec}, then it also satisfies the inequalities of \eqref{prob:mle-ec-ineq}. Thus, the feasible set of the equality-constrained problem is contained inside the feasible set of the inequality-relaxed problem. Therefore, the relaxed problem can never yield a smaller optimal value. The challenge is to show that it also cannot yield a strictly larger one.

\noindent\textit{Step 2. Contradiction assumption.}  
To formalize the recursive structure, the inequality constraints of \eqref{prob:mle-ec-ineq} can now be written compactly as
\[
\bV^n \;\geq\; \cT_{\bbt}[\bV^n], \quad \forall n \in [N].
\]
This notation allows us to leverage standard properties of $\cT_{\bbt}$ to verify the equivalence.  We need the following lemma to support the next steps:

We now suppose that $(\widehat\bbt,\widehat \bV)$ is an optimal solution to the inequality-relaxed problem \eqref{prob:mle-ec-ineq}. By feasibility, we have $\widehat \bV^n \geq \cT_{\widehat\bbt}[\widehat \bV^n]$ for each $n$. Assume, for contradiction, that there exists some observation $n$ and state $\bar s$ such that
\[
\widehat V^n_{\bar s} > (\cT_{\widehat\bbt}[\widehat \bV^n])_{\bar s}.
\]
This means that at $\bar s$ the inequality is strict.

\medskip
\noindent\textit{Step 3. Iterative tightening:}  
To exploit the above  slack, we iteratively apply the Bellman operator to $\widehat \bV^n$:
\[
\widehat \bV^{n,0} = \widehat \bV^n,\qquad 
\widehat \bV^{n,t+1} = \cT_{\widehat\bbt}[\widehat \bV^{n,t}], \quad t=0,1,\ldots
\]
Because $\widehat \bV^n \geq \cT_{\widehat\bbt}[\widehat \bV^n]$, we obtain a monotone sequence
\[
\widehat \bV^{n,0} \;\geq\; \widehat \bV^{n,1} \;\geq\; \widehat \bV^{n,2} \;\geq\;\cdots.
\]
Intuitively, each iteration ``pulls down'' the value function toward the equilibrium fixed point. Importantly, by Lemma~\ref{lm:monotone}, the strict slack at $\bar s$ propagates backwards to all its predecessors, ensuring that the inequality chain is strict at those states as well.

Now, under Assumption~\ref{assump:A1}, we have that  the observed origin $s^n_0$ can reach any state, including $\bar s$. Therefore, the strict inequality at $\bar s$ propagates along some path back to $s^n_0$. This implies that after finitely many iterations $T$, the sequence satisfies
\[
\widehat V^{n,T}_{s^n_0} > \widehat V^{n,T+1}_{s^n_0}.
\]
In words, the iterative tightening procedure eventually decreases the value at the origin.

\noindent\textit{Step 4. Objective improvement and contradiction.}  
Now consider a modified solution in which we replace $\widehat \bV^n$ by $\widehat \bV^{n,T+1}$ while keeping all other $\widehat \bV^{n'}$ fixed. This new solution is still feasible to \eqref{prob:mle-ec-ineq}, because each iterate satisfies $\bV^{n,t} \geq \cT_{\widehat\bbt}[\bV^{n,t}]$. However, the objective strictly improves: since the term $V^n_{s^n_0}$ in the objective has decreased, the value of 
\[
v(\sigma_n \mid \widehat\bbt) - V^n_{s^n_0}
\]
has increased. Thus the modified solution attains a strictly larger objective value than $(\widehat\bbt,\widehat \bV)$. 
This contradicts the assumption that $(\widehat\bbt,\widehat \bV)$ was optimal. Therefore, our assumption of strict slack must be false. It follows that at every optimal solution of \eqref{prob:mle-ec-ineq}, we must have $\widehat \bV^n = \cT_{\widehat\bbt}[\widehat \bV^n]$, i.e., the equalities hold. Hence every optimal solution of the relaxed problem is also feasible for the equality-constrained problem. Combined with Step 1, this proves equivalence.
\end{proof}

To further elucidate the implications of Theorem~\ref{th:ineq-relax}, 
we now provide a series of remarks that analyze specific scenarios in which the equivalence between the equality- and inequality-constrained formulations 
either holds, breaks down, or requires additional structural conditions. 
These discussions aim to underscore the critical role of Assumption~\ref{assump:A1} (network conectivity) in ensuring the validity of the relaxation 
and illustrate the pathological cases that may arise when this assumption is violated, thereby deepening our understanding of the theoretical framework.

\begin{remark}[Necessity of Assumption~\ref{assump:A1}]
\begin{figure}[htb]
\centering
\begin{tikzpicture}[
    >=Latex,
    node distance=2.0cm and 2.5cm,
    every node/.style={font=\small},
    state/.style={circle,draw,minimum size=8mm,inner sep=0pt},
    dest/.style={double, double distance=1pt, circle, draw, minimum size=9mm, inner sep=0pt},
    origin/.style={state, fill=black!5},
    unreachable/.style={state, fill=red!5, draw=red!70!black},
    note/.style={rectangle, rounded corners, draw=black!40, fill=black!2, inner sep=3pt}
]

\node[origin]      (s0) {$s_0$};
\node[state, right=of s0] (s1) {$s_1$};
\node[dest, right=of s1] (d) {$d$};
\node[unreachable, below=of s1] (s2) {$s_2$};

\draw[->, thick] (s0) -- (s1);
\draw[->, thick] (s1) -- (d);
\draw[->, thick] (s2) -- (d);

\draw[red!60, dashed, shorten >=3pt, shorten <=3pt] (s0) .. controls +(0.4,-1.0) and +(-0.4,1.0) .. (s2);
\node[red!70!black, font=\scriptsize, below left=0mm and -1mm of s2] {unreachable from $s_0$};

\end{tikzpicture}
\caption{Counterexample showing the necessity of Assumption~\ref{assump:A1}}
\label{fig:assumptionA1-counterexample}
\end{figure}
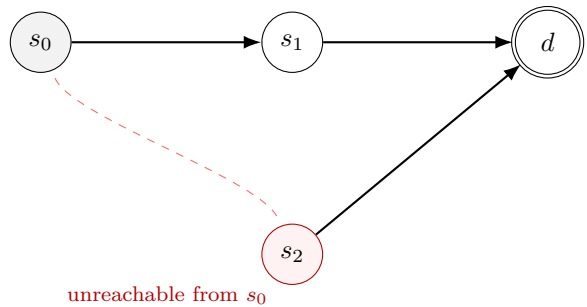
We present an example demonstrating that Assumption~\ref{assump:A1} is necessary to ensure the equivalence between the equality-constrained formulation \eqref{prob:mle-ec} and its inequality relaxation \eqref{prob:mle-ec-ineq}.
Consider a network with three non-destination states $\cS=\{s_0,s_1,s_2\}$ and destination $d$. 
Suppose the observed origin is $s_0$, and the arcs are $s_0 \to s_1 \to d$ and $s_2 \to d$. 
In this network, the path $s_0 \to s_1 \to d$ is valid, but state $s_2$ is not reachable from $s_0$. 

For the inequality-relaxed problem \eqref{prob:mle-ec-ineq}, the Bellman constraints are:
\begin{align*}
V_{s_0} &\;\geq\; \log\!\big( \exp(v(s_1|s_0)+V_{s_1}) \big), \\
V_{s_1} &\;\geq\; \log\!\big( \exp(v(d|s_1)+V_d) \big) = v(d|s_1), \\
V_{s_2} &\;\geq\; \log\!\big( \exp(v(d|s_2)+V_d) \big) = v(d|s_2), \\
V_d &= 0.
\end{align*}
Notice that $V_{s_2}$ never appears on the right-hand side of the recursions for $V_{s_0}$ or $V_{s_1}$, since $s_2$ is not reachable from the origin $s_0$. 
Hence, in the inequality-relaxed formulation, we may set $V_{s_2}$ arbitrarily large (above $v(d|s_2)$) without affecting feasibility or the value of $V_{s_0}$. 
This slack prevents the inequalities from automatically tightening to equalities at the optimum. 

By contrast, if Assumption~\ref{assump:A1} holds, then every state must be reachable from the observed origin $s_0$. 
In that case, the value of $V_{s_2}$ would eventually propagate back to $V_{s_0}$ through the Bellman recursions, and any slack would reduce the objective. 
Thus, under Assumption~\ref{assump:A1}, the inequalities bind at optimality, ensuring equivalence between the inequality and equality formulations.

The example above shows that if Assumption~\ref{assump:A1} does not hold, then solving the inequality-relaxed problem \eqref{prob:mle-ec-ineq} may produce a pair $(\bbt,\bV)$ that does not satisfy the Bellman equations, and therefore is not a valid solution to the original MLE problem. In other words, without this assumption, the relaxation can admit spurious optima that do not correspond to any well-defined recursive logit model. In practice, Assumption~\ref{assump:A1} can usually be enforced through a preprocessing step on the network. Specifically, if some origin states are not connected to parts of the network, one can add artificial arcs that link them to the rest of the graph. These additional arcs can be assigned very large travel times (or disutilities), so that they preserve network connectivity without affecting estimation in any meaningful way. This ensures that the optimization problem is well-posed while leaving the model’s behavioral interpretation intact.
\end{remark}

\begin{remark} [Cycle-free networks]
An important special case arises when the state network is \emph{cycle-free}, i.e., it forms a directed acyclic graph (DAG) from each origin to the destination. In such networks, both the equality-constrained formulation \eqref{prob:mle-ec} and its inequality relaxation \eqref{prob:mle-ec-ineq} are always feasible and admit a solution. 

The intuition is straightforward. Since the network has no directed cycles, the Bellman recursion can be evaluated in a finite number of steps by backward induction starting from the destination $d$, for which $V_d=0$. For each predecessor state, the value function is defined as the log-sum-exp of its outgoing transitions, which is always finite because the recursion terminates at $d$ after finitely many steps. Thus the equalities in \eqref{prob:mle-ec} define a well-posed system with a unique solution for $\{V_s\}_{s\in\cS}$. As a consequence, under Assumption~\ref{assump:A1}, the two formulations \eqref{prob:mle-ec} and \eqref{prob:mle-ec-ineq} are equivalent.  
\end{remark}
\begin{remark}[Case when both formulations are infeasible]
As discussed above, whenever the equality-constrained formulation \eqref{prob:mle-ec} admits a solution, the inequality-relaxed formulation \eqref{prob:mle-ec-ineq} also has at least one feasible solution (with the inequalities binding at optimality). A natural question is whether the converse might hold: if the equality-constrained problem is infeasible due to the presence of cycles or other structural issues, could the inequality relaxation still be feasible? 

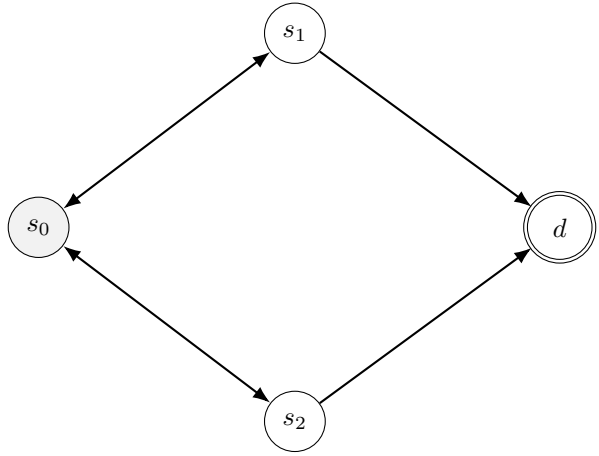
\begin{figure}[htb]
\centering
\begin{tikzpicture}[
    >=Latex,
    node distance=20mm and 28mm,
    every node/.style={font=\small},
    state/.style={circle,draw,minimum size=8mm,inner sep=0pt},
    origin/.style={state,fill=black!5},
    dest/.style={double, double distance=1pt, circle, draw, minimum size=9mm, inner sep=0pt}
]

\node[origin] (s0) {$s_0$};
\node[state, above right=of s0] (s1) {$s_1$};
\node[state, below right=of s0] (s2) {$s_2$};
\node[dest, right=60mm of s0] (d) {$d$};

\draw[<->, thick] (s0) -- (s1);
\draw[<->, thick] (s0) -- (s2);

\draw[->, thick] (s1) -- (d);
\draw[->, thick] (s2) -- (d);


\end{tikzpicture}
\caption{Topology with $s_0\!\leftrightarrow\!s_1$, $s_0\!\leftrightarrow\!s_2$, and terminal arcs $s_1\!\to\!d$, $s_2\!\to\!d$.}
\label{fig:s0s1s2d}
\end{figure}

The following example shows that the answer is negative in general: there exist network topologies and specifications where \emph{both} formulations are infeasible. 
Consider a simple network with three non-destination states $\{s_0,s_1,s_2\}$ and a destination $d$, with arcs
\[
s_0 \leftrightarrow s_1, \qquad s_0 \leftrightarrow s_2, \qquad s_1 \to d, \qquad s_2 \to d.
\]
This network contains two cycles, $s_0 \leftrightarrow s_1$ and $s_0 \leftrightarrow s_2$, while both $s_1$ and $s_2$ also connect directly to the destination. 
We fix $V_d=0$ and let $Z_s \stackrel{def}{=} e^{V_s}>0$. Define
\[
\begin{aligned}
&b_1 = e^{v(s_1|s_0)},\quad b_2 = e^{v(s_2|s_0)},\\
&a_{10}= e^{v(s_0|s_1)},\quad a_{20}= e^{v(s_0|s_2)},\\
&c_1 = e^{v(d|s_1)},\quad c_2 = e^{v(d|s_2)}.
\end{aligned}
\]
The Bellman \emph{equalities} $V=\log\sum\exp(\cdot)$ are then equivalent to the linear system
\[
\begin{aligned}
Z_{s_1} &= a_{10} Z_{s_0} + c_1, \\
Z_{s_2} &= a_{20} Z_{s_0} + c_2, \\
Z_{s_0} &= b_1 Z_{s_1} + b_2 Z_{s_2},
\end{aligned}
\]
which can be reduced to
$\bigl[1 - (b_1 a_{10} + b_2 a_{20})\bigr] \, Z_{s_0} \;=\; b_1 c_1 + b_2 c_2.$
Since $Z_{s_0}>0$ and $b_i,c_i>0$, a finite solution exists if and only if
\[
S \;\stackrel{def}{=}\; b_1 a_{10} + b_2 a_{20} \;<\; 1,
\]
in which case 
\[
Z_{s_0} \;=\; \frac{b_1 c_1 + b_2 c_2}{1-S}.
\]
For the inequality formulation \eqref{prob:mle-ec-ineq}, the same elimination gives
\[
\bigl[1 - S \bigr]\, Z_{s_0} \;\ge\; b_1 c_1 + b_2 c_2,
\]
which is feasible if and only if $S<1$ (by taking $Z_{s_0}$ sufficiently large). Thus, on this topology example, the equality and inequality problems in \eqref{prob:mle-ec} and \eqref{prob:mle-ec-ineq} are feasible together (when $S<1$) and infeasible together (when $S\ge 1$).

To illustrate a case where both formulations can be simultaneously infeasible—meaning that there exists no pair $(\bbt,\bV)$ feasible for \eqref{prob:mle-ec}, nor any pair $(\bbt,\bV)$ feasible for \eqref{prob:mle-ec-ineq}—we construct the following specification. Let $x_0 \in \mathbb{R}^p$ be any nonzero feature vector, and define the utilities linearly in $\bbt$ as
\[
\begin{aligned}
&v(s_1|s_0;\bbt)=\bbt^\top x_0, \qquad &&v(s_0|s_1;\bbt)=\bbt^\top x_0,\\
&v(s_2|s_0;\bbt)=-\,\bbt^\top x_0, \qquad &&v(s_0|s_2;\bbt)=-\,\bbt^\top x_0,\\
&v(d|s_1;\bbt)=c_1,\qquad &&v(d|s_2;\bbt)=c_2,
\end{aligned}
\]
for fixed constants $c_1,c_2\in\mathbb{R}$. In this case,
\[
b_1 a_{10} = e^{2\bbt^\top x_0}, 
\qquad b_2 a_{20} = e^{-2\bbt^\top x_0},
\]
so
\[
S = b_1 a_{10} + b_2 a_{20} 
= e^{2\bbt^\top x_0} + e^{-2\bbt^\top x_0}
 \;\ge 2 \;>\; 1,
\]
for all $\bbt\in\mathbb{R}^p$. 

Thus, under this specification, the condition $S<1$ cannot be satisfied. Hence \emph{both} the equality-constrained problem \eqref{prob:mle-ec} and the inequality-relaxed problem \eqref{prob:mle-ec-ineq} are infeasible for all parameter values $\bbt$. This illustrates that when cycles are present, neither formulation is guaranteed to have a solution, motivating the need for assumptions such as Assumption~\ref{assump:A2}.
\end{remark}



\subsection{Convex Reformulation}
We now present our approach to further reformulate the inequality-constrained formulation in  \eqref{prob:mle-ec-ineq} into a \textit{convex optimization problem} that can be solved efficiently by modern commercial solvers. The key step is to introduce auxiliary variables
\[
Q^n_{ss'} = v(s'|s;\bbt) + V^n_{s'}, \qquad \forall n\in[N], ~ s\in \cS, ~ s'\in A(s).
\]
These variables capture the one-step transition utilities, combining both the deterministic part of the utility and the continuation value. With this definition, the value function recursion can be written in the compact form
\[
V^n_s = \log\!\left(\sum_{s'\in A(s)} \exp(Q^n_{ss'})\right).
\]
This representation makes it possible to directly embed the equilibrium constraints into the optimization problem. Substituting the expression above into \eqref{prob:mle-ec-ineq} leads to the following formulation:
\begin{align}
    \max_{\bbt,\,\bQ} \quad & \sum_{n\in [N]} \left( v(\sigma_n \mid \bbt) - \log\!\Big(\sum_{s'\in A(s^n_0)} \exp(Q^n_{s^n_0 s'})\Big)\right) \label{prob:mle-ec-2}\tag{\sf MLE-EC-2}\\
    \text{s.t.} \quad & Q^n_{ss'} \;\geq\; v(s'|s;\bbt) + \log\!\left(\sum_{t\in A(s')} \exp(Q^n_{s't})\right), \quad \forall n\in[N], ~ s\in \cS, ~ s'\in A(s), \nonumber\\
    & Q^n_{sd} = v(d|s;\bbt), \quad \forall n\in[N],~ s\in\cS. \nonumber
\end{align}
It is important to emphasize that \eqref{prob:mle-ec-2} is a convex optimization problem. The concavity of the objective function follows from the fact that $-\log \sum \exp$ is concave, while $v(\sigma_n \mid \bbt)$ is linear in $\bbt$. On the other hand, the constraints consist of inequalities of the form
\[
Q^n_{ss'} \;\geq\; v(s'|s;\bbt) + \log\!\left(\sum_{t\in A(s')} \exp(Q^n_{s't})\right),
\]
whose right-hand side is convex in $(\bbt,\bQ)$. This guarantees that the feasible set is convex. Together, these two properties imply that the problem is convex and therefore admits a globally optimal solution. This convex formulation can in principle be solved by generic gradient-based convex optimization solvers such as projected gradient descent or quasi-Newton methods (e.g., L-BFGS). These methods guarantee convergence to the global optimum under mild conditions. However, they may suffer from slow convergence, particularly in large-scale networks where the log-sum-exp terms create ill-conditioned subproblems. Moreover, gradient-based approaches typically require careful step-size tuning and may be sensitive to scaling of the problem data. These limitations motivate the use of more specialized reformulations, in particular casting the problem as an exponential cone program, as discussed in the following.

\paragraph{Exponential Cone Reformulation}  
An additional advantage of the reformulation in \eqref{prob:mle-ec-2} is that the log-sum-exp terms can be represented exactly using exponential cone constraints. Recall that the exponential cone is defined as \citep{MOSEKcookbook,DahlAndersenExpCone2022}:
\[
\mathcal{K}_{\exp} = \Big\{(x,y,z)\in \mathbb{R}^3 : y>0, \; y \exp(x/y) \leq z\Big\} \;\cup\; \Big\{(x,0,z) : x \leq 0, \; z \geq 0\Big\}.
\]


The convex formulation in \eqref{prob:mle-ec-2} involves multiple log-sum-exp expressions, each of which can be equivalently modeled using exponential cone constraints. To make this explicit, we introduce additional epigraph variables
\begin{equation}\label{eq:un-s-def}
    u^n_s = \log\!\left(\sum_{t\in A(s)} \exp(Q^n_{st})\right), \qquad \forall n\in[N], ~ s\in\cS,
\end{equation}
which capture the log-sum-exp values at each state. With these variables, problem \eqref{prob:mle-ec-2} can be equivalently written as
\begin{align}
    \max_{\bbt,\,\bQ,\,\bu} \quad & \sum_{n\in [N]} \left( v(\sigma_n \mid \bbt) - u^n_{s^n_0}\right) \label{prob:mle-ec-3}\\
    \text{s.t.} \quad 
    & Q^n_{ss'} \;\geq\; v(s'|s;\bbt) + u^n_{s'}, \quad \forall n\in[N], ~ s\in \cS, ~ s'\in A(s), \nonumber\\
    & Q^n_{sd} = v(d|s;\bbt), \quad \forall n\in[N],~ s\in\cS, \nonumber\\
    & u^n_s \;\geq\; \log\!\left(\sum_{t\in A(s)} \exp(Q^n_{st})\right), \quad \forall n\in[N], ~ s\in\cS. \label{eq:ctr-un-s}
\end{align}
Here we note that the equality definition in \eqref{eq:un-s-def} can be safely replaced by the inequality constraints in \eqref{eq:ctr-un-s}. The reason is that the objective function decreases monotonically in $u^n_{s^n_0}$, so at optimality there is no incentive to make $u^n_s$ larger than the log-sum-exp value. In other words, any slack in \eqref{eq:ctr-un-s} would strictly worsen the objective, and hence all inequalities must bind at the optimum. This guarantees that the relaxation is exact and that \eqref{prob:mle-ec-3} is equivalent to \eqref{prob:mle-ec-2}.

The constraints in \eqref{eq:ctr-un-s} are of the generic form
\[
x \;\geq\; \log\!\left(\sum_{i=1}^m \exp(w_i)\right).
\]
It is standard in conic optimization that such log-sum-exp constraints can be reformulated exactly using exponential cones. Specifically, we have the equivalence
\[
x \;\ge\; \log\!\Big(\sum_{i=1}^m e^{w_i}\Big)
\quad\Longleftrightarrow\quad
\begin{cases}
    (w_i-x,\,1,\,r_i)\in \mathcal{K}_{\exp}, & \forall i, \\[0.5em]
    \sum_{i=1}^m r_i \le 1.
\end{cases}
\]
Indeed, the  condition $(w_i-x,\,1,\,r_i)\in \mathcal{K}_{\exp}$ implies that $r_i \ge e^{w_i-x}$. Summing over $i$ yields
\[
\sum_{i=1}^m r_i \;\ge\; \sum_{i=1}^m e^{w_i-x} \;=\; e^{-x} \sum_{i=1}^m e^{w_i}.
\]
Thus the constraint $\sum_i r_i \le 1$ enforces
$e^{-x} \sum_{i=1}^m e^{w_i} \;\le\; 1,$
which is equivalent to $x \ge \log \sum_i e^{w_i}$. This provides an exact exponential-cone representation of the log-sum-exp inequality. Applying this transformation to \eqref{eq:ctr-un-s}, we obtain the following exponential-cone constraints using auxiliary variables $r^n_{ss'}$:
\begin{align*}
   & (Q^n_{ss'}-u^n_s,\,1,\,r^n_{ss'}) \in \mathcal{K}_{\exp}, 
   && \forall n\in[N],~s\in\cS,~s'\in A(s),\\
   & \sum_{s'\in A(s)} r^n_{ss'} \le 1, 
   && \forall n\in[N],~s\in\cS.
\end{align*}
Collecting everything, the exponential-cone reformulation of \eqref{prob:mle-ec-3} is:

\begin{tcolorbox}[colback=gray!10, colframe=black!60, title={Exponential-Cone Formulation for RL Estimation}, breakable]
\begin{align}
\max_{\bbt,\bQ,\bu,\br}\quad 
& \sum_{n\in[N]} \Big( v(\sigma_n\mid \bbt) - u^n_{s^n_0} \Big) \label{prob:mle-ec-exp}\tag{\sf MLE-EC-EXP}\\
\text{s.t.}\quad 
& Q^n_{ss'} \;\ge\; v(s'|s;\bbt) + u^n_{s'}, 
&& \forall n,~s\in\cS,~s'\in A(s), \nonumber\\
& Q^n_{sd} \;=\; v(d|s;\bbt), 
&& \forall n,~s\in\cS, \nonumber\\
& (Q^n_{ss'}-u^n_s,\,1,\,r^n_{ss'}) \in \mathcal{K}_{\exp}, 
&& \forall n,~s\in\cS,~s'\in A(s), \nonumber\\
& \sum_{s'\in A(s)} r^n_{ss'} \le 1, 
&& \forall n,~s\in\cS. \nonumber
\end{align}
\end{tcolorbox}
Problem \eqref{prob:mle-ec-exp} is a \emph{convex conic program}: the objective is linear in $(\bbt,\bQ,\bu)$, the coupling constraints are affine, and the only nonlinearities arise through memberships in the exponential cone $\mathcal{K}_{\exp}$, which is convex. By construction of the epigraph reformulation, \eqref{prob:mle-ec-exp} is equivalent to \eqref{prob:mle-ec-2}, and hence also equivalent to \eqref{prob:mle-ec-ineq}. Therefore, any global optimum returned by a modern exponential-cone solver (e.g., MOSEK) is a globally optimal estimator for the RL MLE problem.

The main advantage of solving the exponential-cone program \eqref{prob:mle-ec-exp}, compared to generic gradient-based convex optimization methods, is its \textit{numerical robustness} and guaranteed \textit{polynomial-time complexity}. Interior-point solvers \citep{nesterov1994interior,mosek2023} for conic programs can exploit the exact conic structure of the log-sum-exp constraints, avoid issues of ill-conditioning common in gradient descent, and typically converge in far fewer iterations. In addition, solvers such as MOSEK \citep{MOSEKcookbook} handle scaling automatically and return certificates of optimality, whereas first-order methods often require delicate step-size tuning and may converge slowly on large-scale networks. 

\paragraph{Computational Complexity of the Exponential-Cone Reformulation: }
As mentioned earlier, one of the main advantages of the exponential-cone reformulation proposed in~\eqref{prob:mle-ec-exp} 
is that it can be efficiently solved using a modern interior-point method with polynomial-time complexity. In the following, we now delve deeper into this complexity analysis.

Interior-point methods \citep{nesterov1994interior} are the state of the art for solving exponential-cone programs. For a generic conic optimization problem with $n$ decision variables and $m$ exponential-cone constraints, the iteration complexity is 
$\mathcal{O}\!\left(\sqrt{m}\,\log\frac{1}{\epsilon}\right),$
where $\epsilon$ denotes the required accuracy. Each iteration requires forming and solving a Newton system whose dimension is proportional to $n+m$. The cost of solving such a system via direct methods is typically $\mathcal{O}((n+m)^3)$, though structure-exploiting solvers can often reduce this to nearly quadratic complexity. Hence, the overall worst-case complexity of an interior-point solver for the exponential-cone reformulation is 
$\mathcal{O}\!\Big((n+m)^3\sqrt{m}\,\log(1/\epsilon)\Big).$

For problem \eqref{prob:mle-ec-exp}, the problem size scales with the network and the data as follows. Each observation $n$ introduces a set of variables $Q^n_{ss'}$ for every arc $(s,s')$, one variable $u^n_s$ for each state $s$, and auxiliary epigraph variables $r^n_{ss'}$ for each arc. Thus, the total number of decision variables is
$n = \mathcal{O}\big(N(|\cA|+|\cS|)\big),$
where $N$ is the number of observed paths, $|\cA|$ is the number of arcs in the network, and $|\cS|$ is the number of states. The number of exponential-cone constraints is also proportional to $N|\cA|$, since each arc contributes one cone constraint. Therefore, the solver complexity in terms of the problem parameters is on the order of
$\mathcal{O}\!\left( (N(|\cA|+|\cS|))^3 \sqrt{N|\cA|}\,\log(1/\epsilon)\right).$
It is worth emphasizing that gradient-based methods such as projected gradient descent or stochastic gradient descent have a per-iteration cost of $\mathcal{O}(N|\cA|)$, since each iteration requires traversing all observed arcs to compute gradients of the log-sum-exp terms. Their iteration complexity to achieve an $\epsilon$-accurate solution is typically $\mathcal{O}(1/\epsilon)$, leading to an overall complexity of 
$\mathcal{O}\!\left(\frac{N|\cA|}{\epsilon}\right)$ \citep{Bubeck2015}.
In practice, this can result in a very large number of iterations, especially when high accuracy is required or when the problem is ill-conditioned. By contrast, interior-point methods are  usually shown to converge within a modest number of iterations—often, regardless of the target accuracy $\epsilon$, but at the expense of a much higher per-iteration computational cost \citep{BoydVandenberghe2004}.

In summary, the exponential-cone reformulation yields a convex conic program with polynomial-time solvability and strong global optimality guarantees. The theoretical worst-case complexity is cubic in the number of variables and constraints, which can be demanding for very large-scale networks. Nevertheless, for moderate-sized networks, interior-point solvers such as MOSEK are highly effective and often far more reliable than first-order methods, especially when accuracy and robustness are critical.

Compared to the classical NFXP algorithm, solving the exponential-cone formulation offers several important advantages. The exp-cone approach is considerably \textit{more robust with respect to initialization}, since feasibility is directly enforced by convex constraints, whereas NFXP requires carefully chosen initial parameters $\bbt$ to ensure that the Bellman equations admit a solution throughout the estimation process \citep{MaiFrejinger2022}. Moreover, NFXP is computationally demanding because it requires solving multiple systems of linear equations both to evaluate the value function $V$ and to compute its Jacobian, with the number of such systems growing proportionally to the dimension of $\bbt$. In contrast, the computational burden of the exp-cone method depends primarily on the size of the network rather than the number of parameters, making it more scalable in that respect. 

\subsection{Network Trimming for Large-Scale Problems}\label{sec:network-trimming}
While the exponential-cone reformulation is robust and guarantees global optimality, the program \eqref{prob:mle-ec-exp} can become prohibitively large in real-world transportation networks. In such settings, the solver must process a vast number of cone constraints, often resulting in heavy memory consumption and long computation times. To address these challenges, it is useful to apply additional pre-processing steps, such as network reduction or approximate formulations, in order to improve scalability. In the following, we introduce a \emph{network trimming} procedure that helps reduce the problem size while maintaining the statistical validity of the estimation.

The idea is that,
 given a specific origin–destination  pair $(o,d)$, not all states in the network are equally relevant. In particular, states that are far from both the origin and the destination are very unlikely to appear in feasible paths, and thus contribute negligibly to the likelihood. Removing such states from the optimization problem can reduce the size of the exponential-cone program significantly without affecting estimation outcomes. To operationalize this idea, we use the concept of \emph{network flow} under a reference parameter vector $\bbt^0$. We first choose $\bbt^0$ such that the deterministic utilities $v(s'|s;\bbt^0)$ are sufficiently negative (e.g., by scaling travel times or costs), which guarantees that the Bellman equation $\bV = \cT_{\bbt^0}[\bV]$ admits a well-defined solution. Based on this solution, we can compute approximate choice  probabilities at each state $P_{\bbt^0}(s'|s)$. Specifically, to compute flow values from the orginal state $o$,  let $P_{\bbt^0}$ denote the $|\cS|\times|\cS|$ transition matrix with entries
\[
\big(P_{\bbt^0}\big)_{s,s'} = 
\begin{cases}
P_{\bbt^0}(s'|s), & \text{if } s' \in A(s),\\
0, & \text{otherwise}.
\end{cases}
\]
Let $F \in \mathbb{R}^{|\cS|}$ be the vector of flow values, and let $e_o$ be the unit vector corresponding to the origin state $o$. Injecting one unit of flow at $o$ and propagating it through the network yields the fixed-point system
$F = e_o + P_{\bbt^0}^\top F.$
Rearranging gives the linear system
$(I - P_{\bbt^0}^\top) F = e_o.$
Provided the spectral radius of $P_{\bbt^0}$ is strictly less than one, this system admits a unique finite solution \citep{BailComi08}. Each entry $F(s)$ then represents the expected flow through state $s$, starting from origin $o$ under the reference policy $P_{\bbt^0}$.

For the trimming criterion, states with flow values below a prescribed threshold $\varepsilon > 0$ are deemed irrelevant and can be removed from the network, along with their associated arcs:
$\cS_{\text{trim}} = \{ s \in \cS : F(s) \geq \varepsilon \}$. The optimization problem is then restricted to the trimmed network $(\cS_{\text{trim}}, \cA_{\text{trim}})$, where $\cA_{\text{trim}}$ denotes the set of arcs between states in $\cS_{\text{trim}}$. This results in a reduced exponential-cone program with significantly fewer variables and constraints.
It is worth noting that, beyond reducing the size of the network, the trimming procedure preserves the structural property required in our convex reformulation. In particular, the trimmed network $(\cS_{\text{trim}},A_{\text{trim}})$ automatically satisfies Assumption~\eqref{assump:A1} (network connectivity with respect to the origin), even if the original network does not satisfy this assumption. We formalize this in the proposition below.
\begin{proposition}
   The trimmed network $(\cS_{\text{trim}},\cA_{\text{trim}})$ satisfies Assumption~\eqref{assump:A1}.  
   That is, for any state $s \in \cS_{\text{trim}}$, there exists a directed path from the origin $o$ to $s$.
\end{proposition}
\begin{proof}
By construction, a state $s$ is included in $\cS_{\text{trim}}$ if and only if its flow value $F(s)$, obtained from the system 
$(I - P_{\bbt^0}^\top)F = e_o,$
satisfies $F(s) \geq \varepsilon > 0$. Since $F(s)$ represents the expected amount of unit flow arriving at state $s$ when starting from the origin $o$ and propagating along feasible transitions, a strictly positive flow value $F(s) > 0$ implies that there must exist at least one path from $o$ to $s$ with positive probability under the reference policy $P_{\bbt^0}$.  Therefore, for every state $s \in \cS_{\text{trim}}$, the trimming rule ensures the existence of a path from $o$ to $s$. This is precisely the requirement of Assumption~\eqref{assump:A1}. Hence, the trimmed network $(\cS_{\text{trim}},A_{\text{trim}})$ satisfies the connectivity assumption.
\end{proof}

This trimming procedure has two important properties: (i) it preserves all high-probability states and arcs, ensuring that the resulting estimator remains accurate; and (ii) it substantially reduces the size of the convex reformulation in practice, particularly for large networks with many peripheral states. In applications, the choice of the threshold $\varepsilon$ provides a trade-off between computational tractability and approximation fidelity. Our later experimental results will demonstrate that combining the proposed network trimming procedure with the convex reformulation enables efficient estimation of the RL model on a large-scale, real-world network.

\section{Numerical Experiments}\label{sec:experiments}
To evaluate the performance of the proposed convex reformulation for RL model estimation, we focus on two fundamental applications of the recursive logit framework: \emph{route choice modeling}~\citep{FosgFrejKarl13} and \emph{multiple discrete choice modeling}~\citep{tran2024network}. These two domains feature networks with distinct structural characteristics, providing complementary test cases for assessing the effectiveness and generality of our approach. The proposed ECP method is compared against the standard NFXP algorithm~\citep{Rust87}, with the MOSEK conic solver~\citep{mosek2023} used to handle the exponential-cone reformulation.
Although several alternative methods have been proposed as substitutes for NFXP, they are not directly suitable for inclusion in our experiments. For instance, the value decomposition method of~\citet{mai2018decomposition} can substantially accelerate RL model estimation but cannot accommodate general route choice settings with individual-specific attributes. Similarly, the MPEC-based approaches developed for dynamic discrete choice models~\citep{SuJudd2012,SchjerningNotesMPEC} have been shown to be less computationally efficient than the standard NFXP, particularly in large-scale problems \citep{IskhakovEtAl2016}. Consequently, our comparisons focus on the classical NFXP benchmark to ensure both fairness and interpretability in evaluating the advantages of the proposed convex reformulation.

All experiments are 
conducted on a Google Cloud TPU v3-8 server equipped with a 96-core Intel Xeon 2.00 GHz CPU. 
In the following, we present our experimental results for each of the two application domains.

\subsection{Multiple Discrete Choice Problem}
Multiple discrete choice (MDC) problems \citep{Bhat05,Bhat08, tran2024network} are central to a wide range of economic and behavioral applications, as individuals or households often make simultaneous selections of multiple alternatives rather than a single option. Examples include consumers purchasing a basket of goods during a shopping trip, families choosing several entertainment or leisure activities, or commuters selecting a combination of transport modes and transfer options. The analysis of MDC behavior is therefore critical in domains such as marketing, urban planning, and transportation policy, as it enables richer modeling of preferences, substitution patterns, and complementarities among alternatives.

Traditional approaches to MDC within the RUM framework, dating back to \cite{McFa81}, typically model each combination of alternatives as a distinct composite choice. While theoretically sound, this formulation quickly becomes computationally intractable due to the combinatorial explosion in the number of possible choice sets. For instance, when $m$ elemental alternatives are available, the number of potential composite alternatives grows exponentially as $2^m$, rendering classical RUM estimation infeasible even for moderately sized $m$.

To address this scalability barrier, \cite{tran2024network} recently proposed the \emph{Logit-based Multiple Discrete Choice} (LMDC) model, a RUM-based framework that exploits a network-based representation of MDC problems. In LMDC, composite alternatives are represented as feasible paths in a DAG. Each node or link in the DAG corresponds to a decision regarding the inclusion or exclusion of an elemental alternative, and complete paths encode valid multi-item choice bundles. This construction transforms the original combinatorial problem into a RL route choice model on the DAG. Importantly, the recursive structure enables efficient computation of choice probabilities using RL techniques developed in the context of route choice modeling. As a result, the LMDC model preserves the theoretical rigor of the RUM framework while maintaining tractability and scalability in high-dimensional multiple discrete choice settings.

There are two DAG representations proposed for the LMDC problem, as illustrated in Figure \ref{fig:dags}. Each topology has its own properties. BiC-DAG has significantly fewer arcs compared to MuC-DAG, which reduces the graph density but also puts more pressure on individual arcs and causes computational instability.

\begin{figure}
\centering
\begin{subfigure}{0.48\linewidth}
    \includegraphics[width=\textwidth]{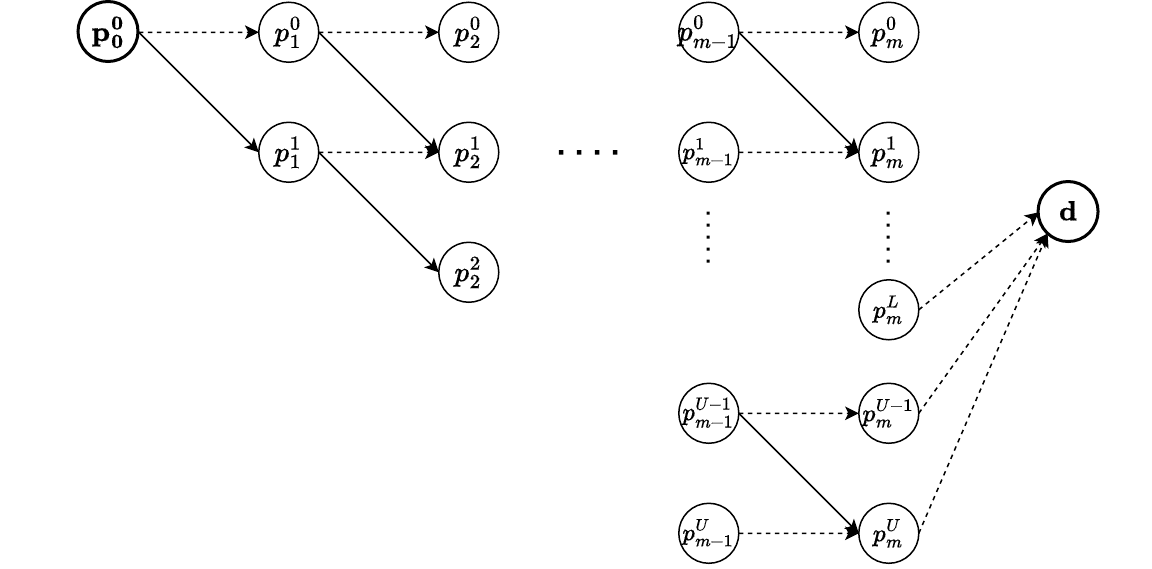}
    \caption{Binary-choice DAG (BiC-DAG)}
\end{subfigure}
\hfill
\begin{subfigure}{0.48\linewidth}
    \includegraphics[width=\textwidth]{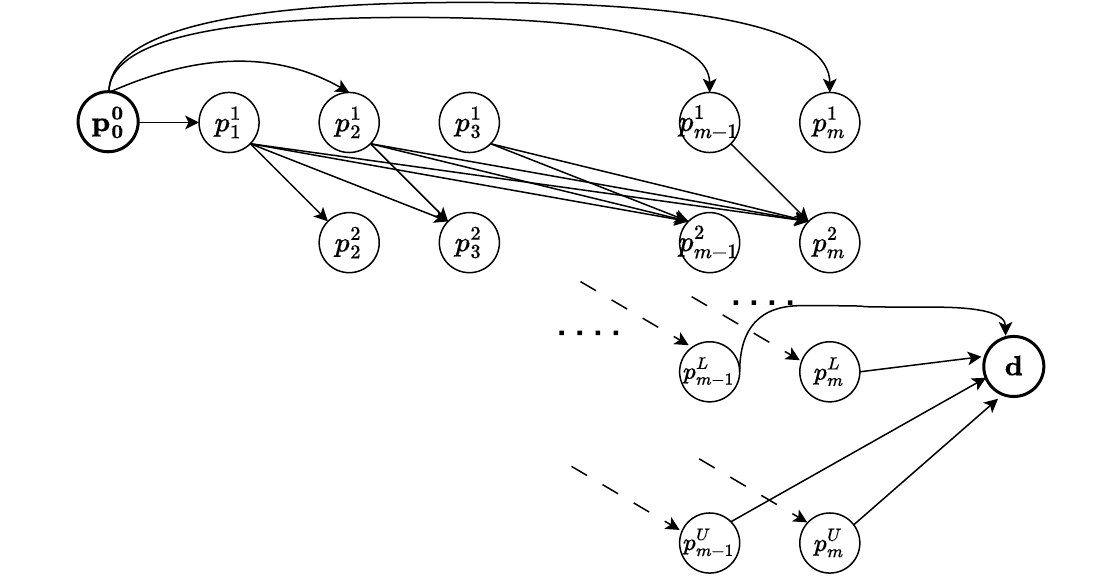}
    \caption{Multi-choice DAG (MuC-DAG)}
\end{subfigure}
\caption{DAG-based representations for LMDC model.}
\label{fig:dags}
\end{figure}

In this experiment, we compare the performance of the proposed \emph{equilibrium-constrained reformulation}, solved using the MOSEK conic optimizer, against the classical NFXP method originally employed in \citet{tran2024network} and prior studies on RL models.
 Recall that the NFXP approach estimates the parameters $\bbt$ by repeatedly solving a system of linear equations via value iteration in an inner loop, while updating $\bbt$ in an outer loop using gradient-based optimization. By contrast, our method reformulates the MLE problem as a convex exponential-cone program, which can be solved in a single stage using interior-point algorithms \citep{fujisawa1997interior}. 
Specifically, we consider the following two estimation methods for comparison:
\begin{itemize}
    \item \textbf{NFXP (Nested Fixed-Point)}:  
    The NFXP algorithm updates the parameter vector $\bbt$ using a quasi-Newton or gradient-based search. At each iteration, for a fixed $\bbt$, the log-likelihood function and its gradient must be evaluated. This requires solving systems of linear equations to recover the value functions, a computationally expensive inner loop. Although both the BiC-DAG and MuC-DAG topologies can be used equivalently in principle, the matrix-based calculations underlying the value iteration are not strongly influenced by graph density. Hence, the structurally denser MuC-DAG is typically the preferred representation for NFXP, as it provides more numerical stability in practice.
    \item \textbf{ECP (Equilibrium-Constrained Program)}:  
    Our equilibrium-constrained reformulation, which can also be expressed as an {exponential-cone program} \citep{chares2009cone}, is solved directly using MOSEK under its default settings. Unlike NFXP, this approach eliminates the need for a nested structure by treating both the parameters $\bbt$ and the value functions as decision variables in a single convex optimization problem. Since the size of the exponential-cone program scales with the number of arcs in the network, and thus depends heavily on graph density, we adopt the BiC-DAG topology in this case to reduce the number of constraints and improve computational efficiency.
\end{itemize}
This comparison underscores a fundamental methodological trade-off. The NFXP algorithm relies on repeated evaluations of an inner fixed-point system, and in practice it may suffer from numerical instability and the absence of guaranteed convergence. By contrast, our conic optimization approach exploits a convex reformulation that ensures global convergence within polynomial time. In the following, we present comparative results on both synthetic and real-world datasets to illustrate these trade-offs in practice.

\subsubsection{Synthetic Data}
The synthetic dataset, together with all settings and instances, is taken from the original paper \cite{tran2024network}. 
Problem instances vary by the number of elemental alternatives $m \in \{5,10,20,30,50\}$ and by the bounds $[L,U]$, 
which represent the minimum and maximum number of alternatives an individual may simultaneously consider. 
The lower bound is fixed as $L=0$ for all instances, allowing individuals to select nothing, 
while the upper bound $U$ is specified for each size: $U=3$ for $m=5$, $U \in \{3,5\}$ for $m=10$, 
$U \in \{5,10\}$ for $m=20$, $U \in \{10,20\}$ for $m=30$, and $U \in \{20,30\}$ for $m=50$. 
The elemental alternative values are identical to those in the original experiment. 
The ``ground-truth'' model is the LMDC framework with parameter vector $\bbt^0 = [-0.5, -0.02, -0.1]$. 
For each instance, two sets of observations are generated: one estimation set of size $1000$ for parameter estimation, 
and one prediction set of size $300$ to evaluate the average log-likelihood achieved by the estimated parameters. 
For comparison, log-likelihood values closer to zero indicate better predictive performance.

It is worth noting that the DAG networks used in this application are cycle-free. 
Consequently, both the equality-constrained formulation \eqref{prob:mle-ec} 
and the exponential-cone reformulation \eqref{prob:mle-ec-exp} always admit a unique optimal solution. 
In particular, solving the value functions in the NFXP framework is always successful under these settings.

\paragraph{Estimation time comparison.}
Since both the NFXP and ECP formulations are theoretically equivalent, and the underlying networks are acyclic, the estimation process always converges successfully. Consequently, both methods yield identical parameter estimates. We therefore focus our comparison exclusively on estimation times. Table~\ref{tab:synth-time} and Figure \ref{fig:synth-time-bar} reports the average and standard deviation of the estimation times obtained by both approaches over 40 independent runs. The results clearly highlight the computational advantages of the proposed ECP method. Although the NFXP approach using the MuC-DAG representation is slightly slower than the BiC-DAG variant---as also observed in \citet{tran2024network}---even the faster NFXP+BiC configuration remains significantly slower than our ECP approach using BiC-DAG. The superior efficiency of ECP in this setting can be attributed to the low-density topology of the BiC-DAG representation, which substantially reduces the size of the optimization problem. Conversely, when the high-density MuC-DAG is employed, the running time of ECP increases notably due to the larger number of arcs and decision states.

\begin{table}[htbp]
    \centering
    \caption{Estimation times (in seconds) for LMDC’s synthetic datasets.}
    \label{tab:synth-time}
    \begin{tabular}{lllll}
        \hline
        \multicolumn{2}{l}{Dataset} &  & \multicolumn{2}{l}{Model} \\ \cline{1-2} \cline{4-5}
        $m$ & $[L,U]$ &  & NFXP & ECP \\ \hline
        5   & [0,3]   &  & $0.09 \pm 0.02$ & $0.02 \pm 0.01$ \\ \hline
        10  & [0,3]   &  & $0.08 \pm 0.01$ & $0.03 \pm 0.00$ \\
        10  & [0,5]   &  & $0.10 \pm 0.01$ & $0.03 \pm 0.00$ \\ \hline
        20  & [0,5]   &  & $0.19 \pm 0.02$ & $0.11 \pm 0.06$ \\
        20  & [0,10]  &  & $0.30 \pm 0.03$ & $0.10 \pm 0.01$ \\ \hline
        30  & [0,10]  &  & $0.57 \pm 0.03$ & $0.17 \pm 0.02$ \\
        30  & [0,20]  &  & $0.99 \pm 0.07$ & $0.29 \pm 0.04$ \\ \hline
        50  & [0,20]  &  & $1.58 \pm 0.12$ & $0.34 \pm 0.04$ \\
        50  & [0,30]  &  & $3.97 \pm 0.30$ & $0.92 \pm 0.18$ \\ \hline
    \end{tabular}
\end{table}

\begin{figure}[htbp]
    \centering
    \includegraphics[width=0.85\linewidth]{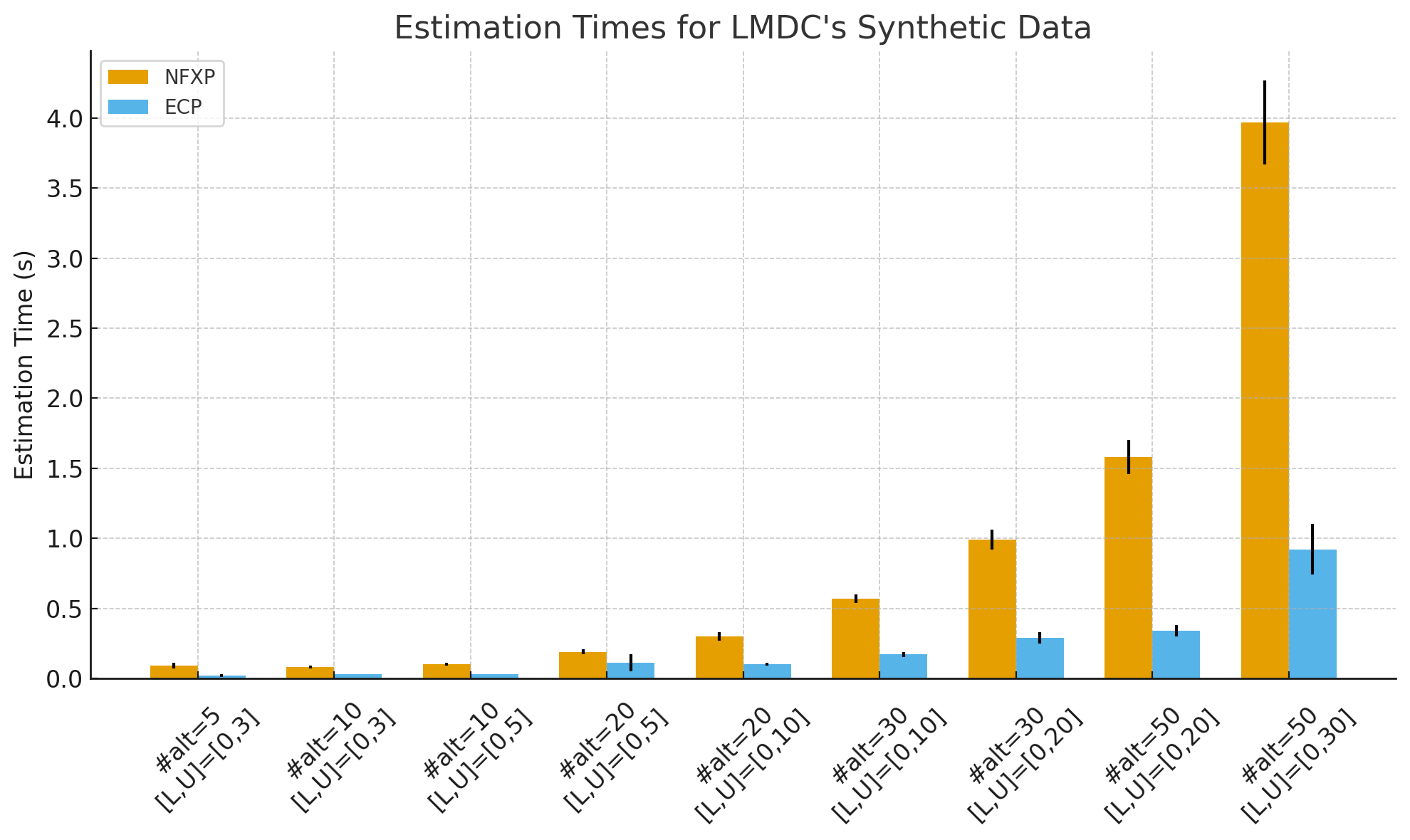}
    \caption{Estimation times (in seconds) for LMDC's synthetic data, comparing NFXP and ECP. 
    Each dataset is identified by the number of alternatives (\#alt) and the corresponding $[L,U]$ range.}
    \label{fig:synth-time-bar}
\end{figure}

\paragraph{Scalability with respect to the number of model parameters.}
We compare the scalability of the proposed ECP and the classical NFXP method in terms of the dimensionality of the parameter vector $\boldsymbol{\beta}$. As illustrated in Figure~\ref{fig:synth-time-params}, the estimation time of NFXP increases almost linearly as the number of parameters grows, whereas ECP remains largely unaffected. This contrast arises from their computational structures: NFXP must repeatedly solve a sequence of linear systems for both value function evaluation and Jacobian computation, with the number of such systems scaling proportionally to $\dim(\boldsymbol{\beta})$. In contrast, ECP incorporates the parameters as additional decision variables within a convex conic program whose size depends primarily on the network topology rather than on $\boldsymbol{\beta}$. Consequently, increases in parameter dimensionality introduce negligible computational overhead for ECP, as the underlying constraint structure and sparsity pattern remain fixed. Empirical results consistently confirm this behavior across all tested datasets: NFXP exhibits a steady increase in runtime with larger parameter sizes, occasionally showing minor fluctuations due to multithreading optimizations in low-level numerical libraries, while ECP maintains stable and consistent computation times. Overall, these findings demonstrate that ECP achieves superior scalability with respect to $\boldsymbol{\beta}$, making it well-suited for high-dimensional estimation problems.

\begin{figure}[htbp]
    \centering
    \includegraphics[width=\linewidth]{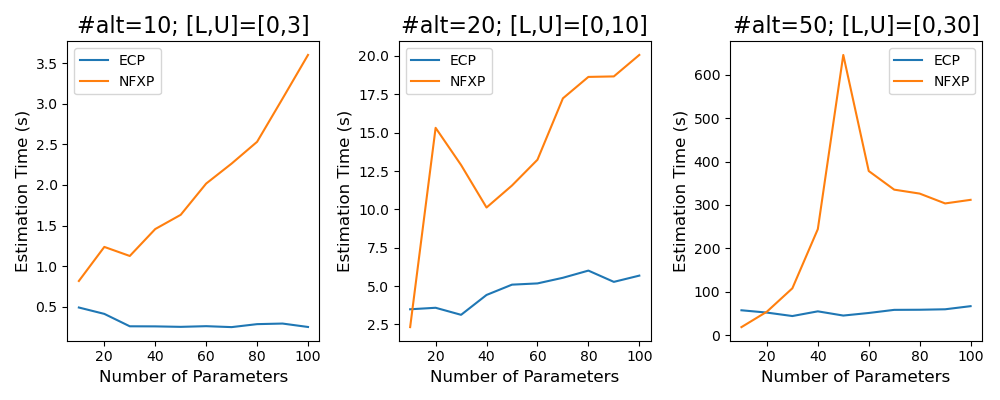}
    \caption{Estimation times (in seconds) for LMDC’s synthetic data as a function of the number of parameters.}
    \label{fig:synth-time-params}
\end{figure}


\subsubsection{Real Data}
\begin{table}[htbp]
    \centering
    \begin{tabular}{l|rrrrc|c}
\hline
\multirow{2}{*}{Model} & \multicolumn{4}{l}{Attributes}      & \multirow{2}{*}{$\mathcal{L}(\hat{\boldsymbol{\beta}})$} & \multirow{2}{*}{\begin{tabular}[c]{@{}c@{}}Time\\ (s)\end{tabular}} \\ \cline{2-5}
                       & Quantity & Price  & Gender & Const  &                                             &                                                                     \\ \hline
NFXP                   & 2.363    & -0.144 & 0.005  & -6.527 & \multicolumn{1}{r|}{-270{,}061}               & \multicolumn{1}{r}{332.72}                                          \\ \hline
ECP                    & 2.362    & -0.144 & 0.005  & -6.527 & \multicolumn{1}{r|}{-270{,}056}               & \multicolumn{1}{r}{27.82}                                           \\ \hline
\end{tabular}
    \caption{Estimation results for the Jewelry-Store dataset.}
    \label{tab:jewelry}
\end{table}
We further evaluate the practical performance of the proposed ECP against the classical NFXP approach using two real-world datasets---the \textit{Jewelry-Store} and \textit{Book-Crossing} datasets---provided in \citet{tran2024network}. Both experiments are conducted under identical experimental settings.  The \textit{Jewelry-Store} dataset contains purchase records from an online jewelry retailer collected over a three-year period. It includes 436 elemental alternatives (products), each characterized by four attributes: \textit{Quantity}, \textit{Price}, \textit{Gender}, and a constant term equal to one across all products. The estimation set consists of 34,291 observations. Table~\ref{tab:jewelry} summarizes the estimation results. The estimated parameters obtained by NFXP and ECP are nearly identical, confirming the theoretical equivalence established in earlier sections. The slight difference in log-likelihood values is attributed to minor numerical rounding effects inherent in the value iteration computations. In contrast, the computational efficiency difference is substantial: ECP achieves convergence in only 28 seconds, compared to 333 seconds for NFXP, representing more than an order-of-magnitude improvement in runtime.

\begin{table}[htbp]
    \centering
    \begin{tabular}{l|rrrrc|c}
\hline
\multirow{2}{*}{Model} & \multicolumn{4}{l}{Attributes}          & \multirow{2}{*}{$\mathcal{L}(\hat{\boldsymbol{\beta}})$} & \multirow{2}{*}{\begin{tabular}[c]{@{}c@{}}Time\\ (s)\end{tabular}} \\ \cline{2-5}
                       & Rating & Popularity & Lifespan & Const  &                                             &                                                                     \\ \hline
NFXP                   & 0.777  & 2.826      & -0.078   & -4.464 & \multicolumn{1}{r|}{-195{,}208}               & \multicolumn{1}{r}{70.57}                                           \\ \hline
ECP                    & 0.777  & 2.826      & -0.078   & -4.464 & \multicolumn{1}{r|}{-195{,}208}               & \multicolumn{1}{r}{7.11}                                            \\ \hline
\end{tabular}
    \caption{Estimation results for the Book-Crossing dataset.}
    \label{tab:book}
\end{table}

The second dataset, \textit{Book-Crossing}, contains four weeks of book purchase histories and user ratings from Amazon. This dataset comprises 100 best-selling books, each described by four attributes: \textit{Popularity}, \textit{Rating}, \textit{Lifespan}, and a constant term. The sample includes 16,989 observations. As shown in Table~\ref{tab:book}, both methods yield identical parameter estimates and log-likelihood values, demonstrating the robustness of the ECP approach. However, ECP again achieves a remarkable computational advantage, completing estimation in approximately 7 seconds compared to more than 70 seconds required by NFXP. This performance gap stems primarily from ECP’s convex reformulation, which eliminates the repeated nested evaluations of the Bellman equations inherent in NFXP. Overall, these real-world experiments confirm that ECP not only matches the statistical accuracy of the classical NFXP framework but also provides substantial gains in computational efficiency, particularly in large-scale settings with many product alternatives.

\subsection{Route Choice Modeling}
The RL model was originally proposed for route choice problems in transportation networks~\citep{Fosgerau2013}. It aims to model how travelers select paths between origin-destination pairs based on factors such as travel time or distance. Traditional multinomial logit models treat entire routes as discrete alternatives; however, this formulation quickly becomes intractable as the number of feasible paths grows combinatorially with network size. The RL model overcomes this limitation by framing route choice as a sequence of link-level decisions, enabling the computation of choice probabilities without explicitly enumerating all possible paths. Despite this theoretical advantage, the estimation of RL models remains computationally demanding. The NFXP algorithm  is the conventional approach for parameter estimation in the route choice context, but as discussed earlier, it frequently fails to recover valid value functions $\mathbf{V}$, particularly when the magnitudes of the parameters $\boldsymbol{\beta}$ are small or when the network exhibits ill-conditioned structures \citep{MaiFrejinger2022}. 




\subsubsection{Synthetic Data}

In this experiment, synthetic traffic networks are constructed using random geometric graph methods - a topology commonly used in graph algorithm benchmarks \citep{penrose2003random}. Two types of transport networks are considered: DAGs and undirected graphs. DAG-based networks have predetermined origin and destination nodes. Undirected graph-based networks are undirected versions of DAG instances (by converting all arcs into undirected edges). Four variants of transport networks are spawned (20, 30, 40 and 50 nodes), with 5 instances for each variant. In the RL model, each link represents a transition between a pair of connected arcs and is characterized by four attributes. The first, $TT(a)$, denotes the travel time on link or route $a$, which is proportional to the length of arc $a$. The second attribute, $LT(a \mid k)$, is a left-turn indicator that equals 1 if the turning angle from node $k$ to arc $a$ lies between $30^{\circ}$ and $150^{\circ}$. Similarly, $RT(a \mid k)$ is a right-turn indicator that equals 1 when the turning angle lies between $-150^{\circ}$ and $-30^{\circ}$. Finally, $UT(a \mid k)$ is a U-turn indicator that equals 1 if the absolute turning angle exceeds $150^{\circ}$ in either direction. These attributes collectively capture key geometric and behavioral features of route choice, allowing the RL model to reflect travelers’ preferences for smoother trajectories and to penalize sharp or reversing turns.

The following specification for the link utility function is employed:
\begin{equation}
    v(a \mid k; \boldsymbol{\beta}) 
    = \beta_{TT} \, TT(a) 
    + \beta_{LT} \, LT(a \mid k) 
    + \beta_{RT} \, RT(a \mid k) 
    + \beta_{UT} \, UT(a \mid k).
\end{equation}
To generate synthetic observations, we specify a ground-truth parameter vector $\widehat{\boldsymbol{\beta}}$ in which all coefficients are negative, reflecting travelers’ general preference to avoid delays and difficult maneuvers. Among these attributes, travel time exerts the strongest negative influence, while turning penalties decrease in magnitude from U-turns to left turns to right turns. Accordingly, the following parameter values are adopted across all experiments:
$\hat{\beta}_{TT} = -4,~  
\hat{\beta}_{LT} = -0.1,~  
\hat{\beta}_{RT} = -0.05,~ 
\hat{\beta}_{UT} = -0.3.$
A total of 3,000 observations are generated for parameter estimation and in-sample prediction. Further details on the data generation process for each experimental scenario are provided in the corresponding subsections.

\begin{table}[htbp]
    \centering
\begin{tabular}{llllllll}
\hline
\multicolumn{2}{l}{Test} &  & \multicolumn{2}{l}{NFXP} &  & \multicolumn{2}{l}{ECP} \\ \cline{1-2} \cline{4-5} \cline{7-8} 
Nodes     & Instance     &  & $\mathcal{L}(\hat{\beta})/N$    & Time (s)    &  & $\mathcal{L}(\hat{\beta})/N$      & Time (s)      \\ \hline
20        & 0            &  & -4.447    & 0.312       &  & -4.447      & 0.063         \\
20        & 1            &  & -4.112    & 0.292       &  & -4.112      & 0.100         \\
20        & 2            &  & -5.088    & 0.314       &  & -5.088      & 0.137         \\
20        & 3            &  & -4.327    & 0.332       &  & -4.327      & 0.084         \\
20        & 4            &  & -3.920    & 0.281       &  & -3.920      & 0.076         \\ \hline
30        & 0            &  & -4.173    & 0.360       &  & -4.173      & 0.290         \\
30        & 1            &  & -6.195    & 0.475       &  & -6.195      & 0.290         \\
30        & 2            &  & -4.294    & 0.363       &  & -4.294      & 0.111         \\
30        & 3            &  & -5.811    & 0.466       &  & -5.811      & 0.248         \\
30        & 4            &  & -5.658    & 0.385       &  & -5.658      & 0.200         \\ \hline
40        & 0            &  & -6.273    & 0.547       &  & -6.273      & 0.408         \\
40        & 1            &  & -5.553    & 0.376       &  & -5.553      & 0.221         \\
40        & 2            &  & -6.362    & 0.507       &  & -6.362      & 0.481         \\
40        & 3            &  & -4.473    & 0.446       &  & -4.473      & 0.442         \\
40        & 4            &  & -5.023    & 0.439       &  & -5.023      & 0.158         \\ \hline
50        & 0            &  & -6.711    & 0.597       &  & -6.711      & 0.543         \\
50        & 1            &  & -5.271    & 0.445       &  & -5.271      & 0.271         \\
50        & 2            &  & -7.293    & 0.675       &  & -7.293      & 0.596         \\
50        & 3            &  & -7.225    & 0.511       &  & -7.225      & 0.283         \\
50        & 4            &  & -5.182    & 0.493       &  & -5.182      & 0.321         \\ \hline
\end{tabular}
    \caption{Estimation results on DAG-based synthetic networks.}
    \label{tab:dag-vi}
\end{table}

\paragraph{Estimation time comparison.} Table~\ref{tab:dag-vi} presents the estimation results obtained on DAG-based synthetic networks. Across all test instances, the NFXP and ECP approaches produce identical log-likelihood values, confirming their theoretical equivalence and the numerical consistency of the proposed equilibrium-constrained formulation. This alignment demonstrates that the ECP method recovers the same optimal solution as the conventional NFXP algorithm. However, the runtime comparison clearly highlights the computational advantage of ECP. For every tested network size, ECP achieves convergence in a fraction of the time required by NFXP—often more than an order of magnitude faster. This improvement stems from the convex reformulation of the estimation problem, which eliminates the need for nested fixed-point iterations inherent to NFXP. Consequently, ECP not only preserves the statistical accuracy of the original recursive logit estimation but also significantly enhances its scalability and numerical stability, particularly for larger and more complex network structures.

\begin{table}[htbp]
    \centering
\begin{tabular}{llllllll}
\hline
\multicolumn{2}{l}{Test} &  & \multicolumn{2}{l}{NFXP} &  & \multicolumn{2}{l}{ECP} \\ 
\cline{1-2} \cline{4-5} \cline{7-8}
Nodes & Instance &  & \shortstack{Successful rate\\ (\%)} & Time (s) &  & \shortstack{Successful rate\\ (\%)} & Time (s) \\ 
\hline
20 & 0 &  & 25.0 & 4.046 &  & 100.0 & 3.460 \\
20 & 1 &  & 37.5 & 3.529 &  & 100.0 & 6.171 \\
20 & 2 &  & 17.5 & 7.651 &  & 100.0 & 7.113 \\
20 & 3 &  & 47.5 & 7.712 &  & 100.0 & 3.625 \\
20 & 4 &  & 7.5  & 10.791 &  & 100.0 & 3.036 \\ 
\hline
30 & 0 &  & 15.0 & 5.859 &  & 100.0 & 17.193 \\
30 & 1 &  & 32.5 & 14.331 &  & 100.0 & 16.906 \\
30 & 2 &  & 10.0 & 4.489 &  & 100.0 & 5.746 \\
30 & 3 &  & 2.5  & 4.190 &  & 100.0 & 11.339 \\
30 & 4 &  & 55.0 & 13.284 &  & 100.0 & 11.699 \\ 
\hline
40 & 0 &  & 60.0 & 6.036 &  & 100.0 & 29.755 \\
40 & 1 &  & 20.0 & 17.850 &  & 100.0 & 27.253 \\
40 & 2 &  & 20.0 & 22.246 &  & 67.5 & 67.100 \\
40 & 3 &  & 70.0 & 8.234 &  & 97.5 & 37.217 \\
40 & 4 &  & 55.0 & 8.639 &  & 100.0 & 13.137 \\ 
\hline
50 & 0 &  & 85.0 & 10.931 &  & 47.5 & 403.627 \\
50 & 1 &  & 45.0 & 18.749 &  & 100.0 & 24.440 \\
50 & 2 &  & 37.5 & 12.979 &  & 97.5 & 472.037 \\
50 & 3 &  & 5.0  & 6.967 &  & 100.0 & 28.690 \\
50 & 4 &  & 77.5 & 7.845 &  & 92.5 & 87.355 \\ 
\hline
\end{tabular}
\caption{Estimation results on undirected cyclic graph-based synthetic networks.}
\label{tab:undirected}
\end{table}

We now compare the performance of the  NFXP and the proposed ECP approaches under undirected cyclic network settings, which are known to be substantially more challenging for RL model estimation~\citep{MaiFrejinger2022}. In such networks, both methods occasionally fail to produce valid value functions and link choice probabilities due to the presence of cycles, which complicates the convergence of the Bellman operator. To overcome this issue when generating synthetic observations, we temporarily convert each undirected graph into a DAGs. Specifically, the conversion process limits the maximum path length to the number of nodes in the original network. For each undirected graph, a corresponding multi-layer DAG is constructed following these rules: each layer contains one duplicate of all nodes (without intra-layer edges), and a node at layer $k$ can only connect to nodes at layer $k+1$ if the two nodes are connected in the original undirected graph. The origin node is positioned in the first layer and the destination node in the final layer.

We evaluate the two methods using two performance metrics: (i) the \emph{successful rate}, defined as the percentage of runs (out of 40) where the RL model is successfully estimated without numerical failure, and (ii) the \emph{estimation time}, computed as the average runtime over successful runs. A method is deemed unsuccessful if it terminates prematurely due to numerical instability (e.g., the solver is forced to stop) or produces invalid log-likelihood values such as NaN or positive values. The comparative results are summarized in Table~\ref{tab:undirected}.

As shown in the table, the proposed ECP method demonstrates a clear advantage over NFXP in terms of both stability and computational efficiency. Across all network sizes, the success rate of ECP remains close to 100\%, whereas NFXP exhibits highly variable and often poor stability, with success rates dropping as low as 2.5\% in some instances. This instability arises because NFXP relies on iterative fixed-point computations of the value function, which are highly sensitive to initialization and scaling in cyclic network structures. In contrast, ECP formulates estimation as a single convex optimization problem, avoiding the nested iterations that typically lead to divergence in NFXP. Regarding computational performance, ECP also achieves competitive or faster runtimes in most cases despite solving a larger conic program. For smaller networks, both methods exhibit comparable speeds; however, as network size grows, NFXP’s runtime increases more erratically due to repeated failures and restarts, while ECP maintains consistent convergence behavior. For example, with 50-node networks, ECP achieves nearly full stability while completing estimation in practical time, whereas NFXP succeeds in less than half the runs and displays large runtime fluctuations. Overall, these results highlight that ECP provides a significantly more robust and reliable estimation framework for RL models in complex cyclic network environments, preserving accuracy while offering superior numerical stability and scalability.



\paragraph{Scalability with respect to the number of parameters.} 
Similar to the MDC experiments discussed earlier, we evaluate the scalability of both the NFXP and ECP approaches with respect to the dimensionality of the parameter vector $\bbt$. The experiments are conducted on both DAGs and undirected cyclic networks. For each combination of graph size and attribute dimensionality, five instances are generated, and each instance produces 3,000 observations using the same data generation procedures described previously for the DAG and undirected cases. To ensure statistical reliability, each configuration is repeated 40 times, and the reported estimation times correspond to the mean runtime over all successful runs.
\begin{figure}[htbp]
    \centering
    \includegraphics[width=\linewidth]{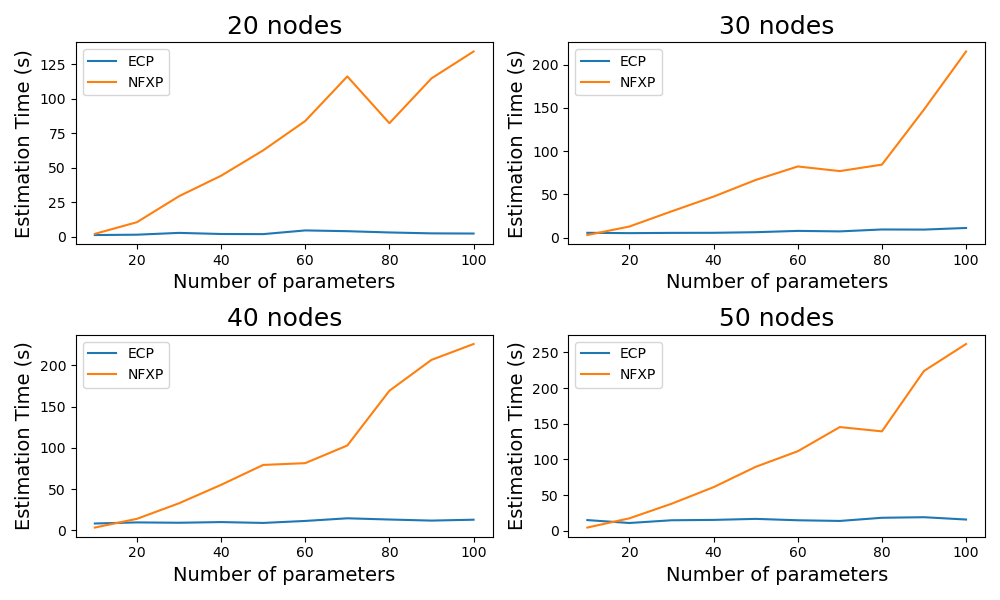}
    \caption{Estimation times for DAG instances by number of parameters.}
    \label{fig:route-time-params}
\end{figure}

For DAG instances, the number of nodes ranges from 20 to 50, and the number of attributes per link varies from 10 to 100. The results, shown in Figure~\ref{fig:route-time-params}, exhibit a consistent pattern with those observed in the LMDC framework: ECP’s estimation time remains largely unaffected by the increase in parameter dimensionality, while NFXP shows an approximately linear growth in runtime. 

\begin{figure}[htbp]
    \centering
    \includegraphics[width=\linewidth]{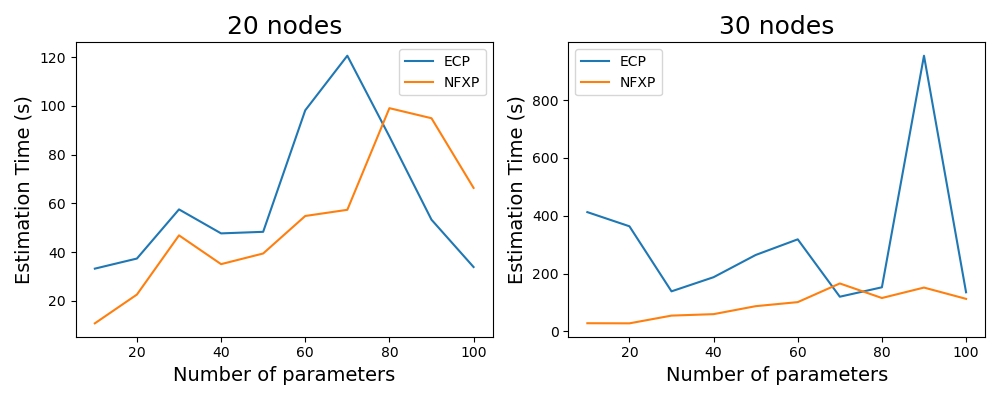}
    \caption{Estimation times for undirected graph instances by number of parameters.}
    \label{fig:route-time-params-undirected}
\end{figure}

For undirected graphs, where the estimation process is considerably more computationally demanding, we restrict the analysis to 20- and 30-node networks. The corresponding results are presented in Figure~\ref{fig:route-time-params-undirected}. Only runs where both NFXP and ECP successfully converged are included, since NFXP frequently encountered numerical instabilities, as demonstrated in earlier experiments. Unlike the DAG case, the results on undirected networks reveal a distinct behavioral pattern. While NFXP continues to scale linearly with the number of parameters, ECP exhibits irregular and less predictable timing behavior, failing to scale smoothly. This degradation in performance is primarily attributed to the underlying cyclic topology: the equilibrium constraints in ECP are more sensitive to feedback loops, which exacerbate the computational burden in solving the conic program. Consequently, ECP’s runtime grows sharply as network complexity increases, even though it still achieves higher stability and convergence reliability compared to NFXP.

In summary, ECP maintains near-constant scalability with respect to parameter dimensionality in DAG settings but suffers in cyclic graphs due to the structural complexity of undirected networks. Conversely, NFXP demonstrates better time scalability in cyclic settings but remains less stable overall. These findings collectively highlight the trade-off between stability and computational efficiency across different network topologies.

\subsubsection{Real-world Large-scale Traffic Network}
We evaluate the performance of the proposed method using a real-world dataset collected from the Borlänge transportation network in Sweden. This dataset has been widely employed in previous studies on route choice modeling~\citep{FosgFrejKarl13,MaiFosFre15,mai2018decomposition}. The underlying network consists of $3{,}077$ nodes and $7{,}459$ links. Since traffic conditions in this network are generally uncongested, link travel times can be treated as static and deterministic.  The sample includes $1{,}832$ observed trips, each represented by a simple path containing at least five links. In total, the dataset covers $466$ destinations, $1{,}420$ distinct origin--destination  pairs, and more than $37{,}000$ link-level choices. The instantaneous link utility function is specified as:
\begin{equation}\nonumber
    v(a \mid k; \boldsymbol{\beta}) 
    = \beta_{TT} \, TT(a)
    + \beta_{LT} \, LT(a \mid k)
    + \beta_{LC} \, LC
    + \beta_{UT} \, UT(a \mid k),
\end{equation}
where $TT(a)$ denotes the travel time on link $a$, $LT(a \mid k)$ is a left-turn indicator that equals one if the transition from link $k$ to link $a$ involves a left turn, $LC$ is a constant attribute equal to one for all links, and $UT(a \mid k)$ is a U-turn indicator for transitions that reverse direction. The topology of the Borlänge transportation network derived from this dataset is illustrated in Figure~\ref{fig:borlange}.
\begin{figure}[htbp]
    \centering
    \includegraphics[width=0.6\linewidth]{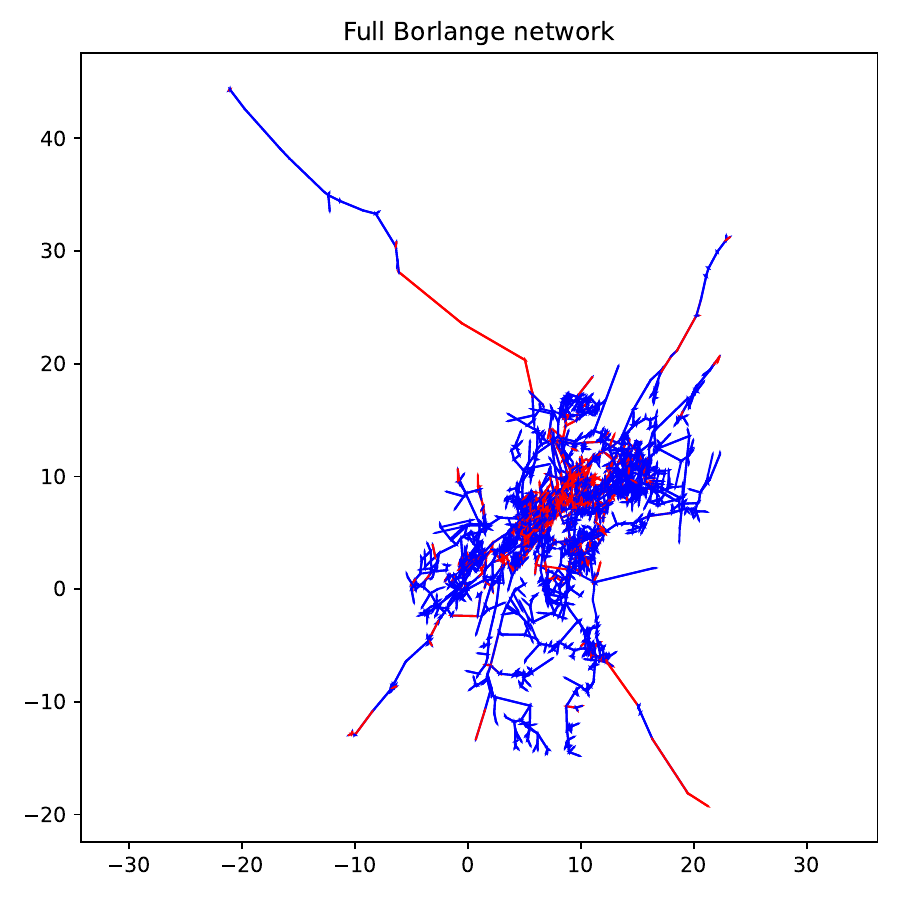}
    \caption{Borlange traffic network. Edges in red are ones directly lead to one of 466 destinations.}
    \label{fig:borlange}
\end{figure}

For the real-world Borlänge network, the large number of links results in an extremely high-dimensional ECP formulation. Consequently, solving the full-scale conic program directly can be computationally demanding, requiring substantial memory resources and often leading to numerical instability. To mitigate these issues, we apply the network trimming procedure introduced in Section~\ref{sec:network-trimming} to reduce the network size before full estimation. The main objective of this preprocessing step is to obtain a smaller yet representative subnetwork in which the ECP can be efficiently solved. By estimating the RL model on the trimmed network, we aim to recover a set of parameter estimates that closely approximate those of the full-scale model. These estimates serve as high-quality initial values for the NFXP algorithm, which can then be applied to the complete network. Starting from parameters near the true optimum significantly enhances the numerical stability of NFXP and reduces the likelihood of divergence in the value function iteration. In essence, this two-stage estimation strategy leverages the robustness and convexity properties of ECP to provide a reliable initialization for NFXP, thereby combining the strengths of both methods: ECP ensures stable convergence in reduced networks, while NFXP refines the estimates to achieve final accuracy on the full network.



\begin{table}[htbp]
    \centering
\begin{tabular}{l|lllll|l}
\hline
\multirow{2}{*}{Model} & \multicolumn{4}{l}{Attributes}                                                    & \multirow{2}{*}{$\mathcal{L}(\hat{\beta})/N$}     & \multirow{2}{*}{Time(s)} \\ \cline{2-5}
                       & $\hat{\beta}_{TT}$ & $\hat{\beta}_{LC}$ & $\hat{\beta}_{LT}$ & $\hat{\beta}_{UT}$ &                                                 \\ \hline
NFXP                     & -1.93              & -6.10              & -2.12              & -1.53              & 10.89     & 46.05                                          \\ \hline
ECP                 & -1.65              & -0.48              & -0.89              & -3.81              & -2.19       & 362.83                                            \\ \hline
\end{tabular}
    \caption{Estimation results on the trimmed Borlange networks.}
    \label{tab:trimmed}
\end{table}

We apply the network trimming procedure described in Section~\ref{sec:network-trimming}, where the threshold parameter $\epsilon$ is chosen such that approximately $90\%$ of the links are removed from the original network. This produces one subnetwork for each origin--destination pair, with a size roughly one-tenth that of the full network. We then estimate the RL model on each subnetwork using both the NFXP and ECP methods. \textit{To ensure a fair comparison, we deliberately exclude stabilization heuristics commonly employed in prior RL studies (e.g., manually replacing invalid values with finite approximations during value function computation).} 
The estimation results obtained from both methods on the trimmed networks are summarized in Table~\ref{tab:trimmed}. In terms of computation time, NFXP demonstrates a clear advantage, requiring only $46$ seconds compared to $363$ seconds for the convex ECP approach. However, NFXP suffers from numerical instability, as indicated by the positive log-likelihood values, while the ECP estimates remain numerically stable and yield parameter values consistent with those reported in the original study by~\citet{Fosgerau2013}. These findings are consistent with our earlier observations from synthetic experiments on undirected graphs, further confirming the robustness of ECP under challenging numerical conditions.




\begin{table}[htbp]
    \centering
\begin{tabular}{l|lllll}
\hline
\multirow{2}{*}{Methods} & \multicolumn{4}{l}{Attributes}                                                    & \multirow{2}{*}{$\mathcal{L}(\hat{\beta})/N$} \\ \cline{2-5}
                       & $\hat{\beta}_{TT}$ & $\hat{\beta}_{LC}$ & $\hat{\beta}_{LT}$ & $\hat{\beta}_{UT}$ &                                                 \\ \hline
                       NFXP on full network                  & \_              & \_              & \_              & \_              & \_                                          \\ \hline
NFXP on trimmed networks + NFXP                     & -1.59              & -0.94              & -1.35              & -1.58              & -4.37                                            \\ \hline
ECP on trimmed networks + NFXP                  & -2.49              & -0.41              & -0.93              & -4.46              & -3.44                                            \\ \hline
\end{tabular}
    \caption{Estimation results on full networks with different estimation approaches.}
    \label{tab:full}
\end{table}

Table~\ref{tab:full} summarizes the estimation results of the RL model on the full Borlänge network using three different estimation strategies, namely: 
(i) direct estimation using NFXP on the full network, 
(ii) NFXP applied to the full network with initialization from parameters estimated on trimmed subnetworks using NFXP, and 
(iii) NFXP applied to the full network with initialization from parameters estimated on trimmed subnetworks using the ECP approach. 
For the first approach, the initial parameter vector is set to $-1.5$ for all coefficients $\beta$, following the setup used in~\citet{FosgFrejKarl13}.
 As shown in the table, direct estimation via NFXP on the full network fails to converge, as indicated by the missing parameter values. This outcome is expected given the extremely large number of links and the inherent instability of the nested fixed-point approach when applied to high-dimensional, cyclic networks. The fixed-point iterations required by NFXP often diverge or return infeasible value functions, leading to numerical failure.

In contrast, the two-stage approaches—where model estimation is first performed on trimmed subnetworks—yield stable results. When NFXP is initialized with parameters estimated from trimmed networks (second row), the algorithm successfully converges on the full network, producing reasonable parameter estimates and a normalized log-likelihood of $-4.37$. The signs and magnitudes of the estimated coefficients are consistent with behavioral expectations: all parameters are negative, reflecting aversion to travel time, turning movements, and U-turns, with travel time ($\hat{\beta}_{TT}$) and U-turns ($\hat{\beta}_{UT}$) having the strongest influence.

The hybrid approach, where ECP is first applied on trimmed subnetworks before refinement with NFXP on the full network, achieves the best balance between numerical stability and model fit. The resulting coefficients remain behaviorally plausible and closer in magnitude to those reported in previous literature~\citep{Fosgerau2013}, while the corresponding log-likelihood ($-3.44$) indicates a superior fit compared to the purely NFXP-based initialization. This improvement demonstrates that ECP provides more reliable and well-conditioned initial estimates, enabling NFXP to converge more efficiently to a high-quality local optimum.

Overall, these results reinforce the advantages of using ECP as a preprocessing step for large-scale RL model estimation. Direct NFXP estimation remains impractical for complex real-world networks, but the combination of ECP-based initialization and subsequent NFXP refinement offers a computationally feasible and statistically robust alternative for achieving accurate model parameters.

Here we note that the estimation of the NRL model is not included in the present experiments. This omission is primarily due to methodological and computational considerations. First, the proposed ECP framework is not directly applicable to NRL under the general setting, as the introduction of nested structures breaks the global convexity of the underlying value function formulation. In particular, the equilibrium constraints that make ECP tractable for standard RL rely on a single-level, additive utility structure, which no longer holds in the presence of hierarchical or correlated error components inherent to NRL. Second, the classical NFXP estimation approach for NRL suffers from the same numerical and convergence issues as those encountered in the standard RL model, such as instability in the value function computation and divergence during the inner fixed-point iteration. In practice, researchers often use the simpler RL model—estimated via NFXP—to obtain a reasonable set of initial parameter values before proceeding with NRL estimation on the same network~\citep{MaiFrejinger2022,MaiFosFre15}. Therefore, including NRL results here would not yield additional insights beyond those already demonstrated for RL, as both methods share similar computational challenges while the ECP framework, by design, cannot be extended to NRL in a straightforward manner. 

\section{Conclusion and Future Work}\label{sec:concl}
This paper proposed an equilibrium-constrained optimization (ECP) framework for estimating RL models, offering a convex and numerically stable alternative to the classical NFXP approach. Theoretical analysis established the equivalence between ECP and NFXP, while experiments on synthetic and real transportation networks demonstrated that ECP achieves similar log-likelihood performance with improved numerical stability and scalability. Empirical results revealed that ECP performs particularly well on directed acyclic networks, where it maintains stable convergence and near-constant scalability as parameter dimensionality grows. In contrast, NFXP remains faster in small-scale problems but exhibits instability and divergence in complex or cyclic networks. When used together, ECP provides reliable initialization that significantly enhances NFXP convergence on large-scale networks.

Future research will focus on extending the ECP formulation to handle the NRL model and other hierarchical structures \citep{Mai_RNMEV}, as well as developing decomposition-based or distributed algorithms to further enhance scalability. In addition, applying the ECP framework beyond transportation—such as to facility location and assortment optimization problems \citep{mai2020multicut,Talluri2004revenue}—represents a promising direction for expanding its applicability.


\bibliographystyle{plainnat_custom}
\bibliography{refs}

\pagebreak

\appendix

\clearpage 

\begin{center}
    {\Huge Appendix}
\end{center}

\section{Extensions to NRL Model}
In this section, we discuss how our approach can be extended to the nested RL model \citep{MaiFosFre15}, which was introduced to account for correlations between path utilities in the network.
The nested RL (NRL) model \citep{MaiFosFre15} is an extension of the standard RL framework designed to capture correlation structures among alternatives in the network. While the RL model assumes that the i.i.d.\ type-I extreme value errors imply independence of irrelevant alternatives (IIA) across all link or node choices, this assumption may be unrealistic in networks where certain paths share common subroutes. The NRL model relaxes the IIA limitation by allowing the scale parameters $\mu$ to vary across states, thereby capturing heterogeneous substitution patterns within different parts of the network. This added flexibility makes the NRL a more behaviorally consistent representation of route choice, particularly in settings with substantial path overlap, and empirical studies have shown that it often provides a better fit than the standard RL model. However, this flexibility comes at a higher computational cost: unlike the standard RL model, where the value function can be obtained by solving a linear system, the NRL requires iterative procedures (e.g., value iteration) to compute the value function, making estimation more demanding in practice.

\paragraph{Model Formulation:} The NRL model generalizes the standard RL framework by allowing the scale parameter of the extreme value shocks to vary across states. Specifically, let $\mu_s > 0$ denote the scale parameter associated with state $s$. The random utility of transitioning from state $s$ to $s'$ is
$u(s'|s) \;=\; v(s'|s;\bbt) + \mu_s \epsilon(s'),$
where $\epsilon(s')$ are i.i.d.\ type-I extreme value errors. Under this specification, the Bellman recursion becomes
\begin{equation}
    V(s) =
    \begin{cases}
        0, & s = d, \\[0.4em]
        \mu_s \log\!\Big(\sum_{s' \in A(s)} 
        \exp\!\Big(\tfrac{1}{\mu_s}\big(v(s'|s;\bbt)+ V(s')\big)\Big)\Big), & s \in \cS.
    \end{cases}
    \label{eq:nrl-bellman}
\end{equation}
At each state $s$, the conditional choice probability of moving to $s'$ is
\[
P(s'|s;\bbt) \;=\;
\frac{\exp\!\big(\tfrac{1}{\mu_s}(v(s'|s;\bbt) + V(s'))\big)}
{\sum_{a \in A(s)} \exp\!\big(\tfrac{1}{\mu_s}(v(a|s;\bbt) + V(a))\big)}.
\]
For an observed path $\sigma_n=(s_0^n,\ldots,s_{T_n}^n=d)$, the log-likelihood is
\begin{align}
\log P(\sigma_n \mid \bbt)
&= \sum_{t=0}^{T_n-1} \Bigg\{ 
\frac{1}{\mu_{s_t^n}} \big( v(s_{t+1}^n|s_t^n;\bbt) + V(s_{t+1}^n) \big)
- \frac{1}{\mu_{s_t^n}} V(s_t^n) \Bigg\}.
\label{eq:nrl-loglik}
\end{align}
Aggregating over all observed paths yields the NRL maximum likelihood estimator.

\paragraph{Equilibrium-constrained reformulation.}
Analogous to the RL case, the NRL estimation can be posed as an optimization with equilibrium constraints:
\begin{align}
\max_{\bbt,\bV,\bmu}\quad
& \sum_{n=1}^N \sum_{t=0}^{T_n-1} 
\left\{\frac{1}{\mu_{s_t^n}} \big(v(s_{t+1}^n|s_t^n;\bbt) + V^n_{s_{t+1}^n}\big) 
- \frac{1}{\mu_{s_t^n}} V^n_{s_t^n} \right\}
\tag{\sf NRL-MLE-EC}\label{prob:nrl-mle-ec}\\
\text{s.t.}\quad
& V^n_s \;=\; \mu_s \log\!\Big(\sum_{a\in A(s)} 
\exp\!\Big(\tfrac{1}{\mu_s}\big(v(a|s;\bbt)+ V^n_a\big)\Big)\Big),
&& \forall n,~ s\in\cS, \nonumber\\
& V^n_d = 0, && \forall n.\nonumber
\end{align}
Following our approach for the standard RL model, the next step is to relax the equalities in \eqref{prob:nrl-mle-ec} to inequalities. The following proposition shows that this relaxation is \emph{exact} only under a structural condition on the state-dependent scales.

\begin{proposition}\label{prop:nrl-ineq}
Consider the inequality relaxation of \eqref{prob:nrl-mle-ec}:
\begin{align}
\max_{\bbt,\bV,\bmu}\quad
& \sum_{n=1}^N \sum_{t=0}^{T_n-1} 
\left\{\frac{1}{\mu_{s_t^n}} \big(v(s_{t+1}^n|s_t^n;\bbt) + V^n_{s_{t+1}^n}\big) 
- \frac{1}{\mu_{s_t^n}} V^n_{s_t^n} \right\}
\tag{\sf NRL-MLE-EC-IEQ}\label{prob:nrl-mle-ec-ieq}\\
\text{s.t.}\quad
& V^n_s \;\geq\; \mu_s \log\!\Big(\sum_{a\in A(s)} 
\exp\!\Big(\tfrac{1}{\mu_s}\big(v(a|s;\bbt)+ V^n_a\big)\Big)\Big),
&& \forall n,~ s\in\cS, \nonumber\\
& V^n_d = 0, && \forall n.\nonumber
\end{align}
Suppose Assumption~\ref{assump:A1} (reachability/connectivity) and Assumption~\ref{assump:A2} (existence of a feasible solution to \eqref{prob:nrl-mle-ec}) hold, and moreover the scale parameters satisfy the forward–monotonicity condition
\[
\mu_{s} \;\ge\; \mu_{s'}\qquad \text{for every arc } (s,s') \in E \text{ (i.e., } s'\in A(s)\text{)}.
\tag{$\star$}\label{eq:mu-monotone}
\]
Then the relaxation \eqref{prob:nrl-mle-ec-ieq} is \emph{equivalent} to \eqref{prob:nrl-mle-ec}: at every optimal solution of \eqref{prob:nrl-mle-ec-ieq}, all inequalities bind and the Bellman equalities hold for every state.
\end{proposition}

\begin{proof}[Proof (brief)]
Fix a path $\sigma_n=(s_0^n,\ldots,s_{T_n}^n=d)$. Collect the coefficients of each $V^n_s$ in the objective of \eqref{prob:nrl-mle-ec-ieq}. Each time $s$ appears as a \emph{current} state $s_t^n$ it contributes $-\tfrac{1}{\mu_s}$; each time $s$ appears as a \emph{next} state $s_{t+1}^n$ from some predecessor $p=s_t^n$ it contributes $+\tfrac{1}{\mu_p}$. For interior visits, every entry into $s$ (as a next state) is paired with a departure from $s$ (as a current state). Under \eqref{eq:mu-monotone}, for every incoming arc $(p,s)$ we have $\tfrac{1}{\mu_p}\le \tfrac{1}{\mu_s}$. Hence, for each interior visit, the net coefficient on $V^n_s$ is
\[
+\tfrac{1}{\mu_p} - \tfrac{1}{\mu_s} \;\le\; 0,
\]
with strict negativity whenever $\mu_p>\mu_s$. Aggregating over all visits shows that, for every $s\in\cS$, the total coefficient of $V^n_s$ in the objective is \emph{non-positive} for internal states, and negative for the original states.

Since the objective function is non-increasing in every $V^n_s$ and strictly decreasing in the origin value $V^n_{s^n_0}$, we can argue as in the RL case that any slack in an inequality constraint at state $s$ can be eliminated. In particular, if $V^n_s$ is strictly larger than the right-hand side of its constraint, then decreasing $V^n_s$ improves the objective, while feasibility is preserved thanks to the monotonicity of the log-sum-exp operator. Moreover, by the connectivity condition (Assumption~\ref{assump:A1}) and the standard monotonicity of the Bellman operator, lowering $V^n_s$ does not force any predecessor’s constraint to become violated. Hence, we may continue decreasing $V^n_s$ until its corresponding constraint becomes active. 
Observe that if $V^n_s$ is strictly decreased, then the value at a predecessor state $s'$ 
can also be strictly decreased. By applying this argument iteratively and invoking the 
connectivity assumption (Assumption~\ref{assump:A1}), it follows that the value at the 
initial state $V^{n}_{s^n_0}$ will also be strictly decreased. Consequently, at the optimum, 
all inequalities must bind, ensuring that the Bellman equations are satisfied with equality 
throughout the network.
 Finally, Assumption~\ref{assump:A2} guarantees the existence of a feasible fixed point, ensuring that the process converges to a valid optimal solution of \eqref{prob:nrl-mle-ec}, thereby establishing the claimed equivalence.
\end{proof}

\begin{remark}[Implications of Condition \eqref{eq:mu-monotone}]
    Condition \eqref{eq:mu-monotone} is a \emph{sufficient} (not necessary) requirement ensuring the objective is (componentwise) non-increasing in $\bV$ and therefore forces epigraph inequalities to bind. When \eqref{eq:mu-monotone} fails, some $V^n_s$ can enter the objective with a \emph{positive} net coefficient (due to predecessors with smaller $\mu$), so slack need not be eliminated at optimality, and the inequality relaxation may strictly enlarge the feasible/optimal set; in that case, the exact convex/exponential-cone reformulation used for RL does not carry over to NRL.

We further note that in the case of a network containing cycles, condition~\eqref{eq:mu-monotone} implies that the scale parameters $\mu_s$ must be identical across all states belonging to a cycle. This is because the monotonicity condition requires $\mu_{s'} \geq \mu_s$ for every successor $s' \in A(s)$. Traversing along the edges of a cycle $s_1 \to s_2 \to \cdots \to s_k \to s_1$ therefore yields the chain of inequalities
$\mu_{s_2} \;\geq\; \mu_{s_1},\quad
\mu_{s_3} \;\geq\; \mu_{s_2}, \;\;\ldots,\quad
\mu_{s_1} \;\geq\; \mu_{s_k}.$
Combining these inequalities forces equality throughout, i.e.,
$\mu_{s_1} = \mu_{s_2} = \cdots = \mu_{s_k}.$
Hence, within any cycle, the scale parameter must remain the same. This result highlights that heterogeneity of $\mu_s$ can only occur across acyclic portions of the network, while cycles impose a strict homogeneity constraint on the scale parameters.

The assumption that $\mu_{s} \geq \mu_{s'}$ for every successor $s' \in A(s)$ is not part of the standard definition of the NRL model, but it can be justified from a behavioral perspective. Recall that the parameter $\mu_s$ controls the scale of the random utility terms: a larger $\mu_s$ indicates greater influence of the random shocks relative to the deterministic utilities, and thus a higher degree of randomness in the choice at state $s$. By requiring that $\mu_{s} \geq \mu_{s'}$ for every successor $s'$, we are essentially assuming that the traveler is more uncertain about decisions made at earlier stages of the trip, while becoming increasingly certain as they approach the destination. This monotonic decrease in uncertainty along the journey is behaviorally intuitive, since route choices tend to become more predictable and focused as the traveler gets closer to completing the trip.
\end{remark}
\begin{remark}[Convexity, trade-offs, and practical strategies in NRL estimation]
We can draw several important implications from the above discussion.  
First, the equilibrium-constrained problem in \eqref{prob:nrl-mle-ec-ieq} is not convex in the joint variables $(\bbt,\bV,\bmu)$. Convexity is preserved only when the scale parameters $\bmu$ are fixed, in which case the problem reduces to a structure similar to the standard RL model. This observation explains why the NRL model is substantially more challenging to estimate.  

 Second, although our main result for the NRL model is negative—namely, that the MLE problem cannot in general be reformulated as a convex optimization problem—this insight highlights the fundamental difficulty of solving the NRL compared to the RL. The RL benefits from convexity and admits an exact exponential-cone reformulation, while the NRL sacrifices this tractability to achieve greater behavioral flexibility.  

 Third, this exposes a fundamental trade-off: the NRL relaxes the IIA property and provides a richer behavioral representation by accounting for correlation among overlapping routes, but this comes at the cost of losing the convex structure that underpins efficient conic optimization methods. In other words, the gain in behavioral realism entails a loss of computational convenience.  

Finally, our convex reformulation for the RL model remains practically useful in the context of the NRL. Specifically, one can first estimate the RL counterpart (where all scale parameters are assumed equal) efficiently using the exponential-cone reformulation, and then use the resulting parameter estimates as high-quality initial values to warm-start the NFXP algorithm for the NRL model. This hybrid strategy can mitigate the sensitivity and instability of direct NFXP estimation in the NRL setting, and hence represents a pragmatic way to leverage our approach even when convexity is not available.  
\end{remark}

Beyond the standard RL and NRL formulations, there exists another variant known as the \emph{Recursive Network Multivariate Extreme Value} (RN-MEV) model~\citep{Mai_RNMEV}, which introduces an even more flexible correlation structure among alternatives. However, the RN-MEV model shares a similar recursive structure with the NRL and can, in fact, be reformulated as an equivalent NRL model defined over an expanded or extended network. Consequently, it inherits the same limitations when applying the proposed ECP approach: the loss of global convexity in the value function and the inability to express the estimation problem as a tractable conic program. Addressing these challenges and extending the ECP framework to accommodate correlated network structures such as RN-MEV remains an open avenue for future research.

\end{document}